%% file: main.tex
\documentclass[nolineno,a4paper,USenglish,cleveref,autoref,thm-restate]{socg-lipics-v2021}

\pdfoutput=1 
\hideLIPIcs  

\usepackage[group-separator={,},output-decimal-marker={.}]{siunitx}
\usepackage{nicefrac}
\usepackage{complexity}
\usepackage{microtype}

\usepackage[]{algorithm2e}
\usepackage{longtable}
\usepackage[commandnameprefix=always,final]{changes}

\crefname{equation}{Equation}{Equations}

\title{Exact Algorithms for Minimum Dilation Triangulation}

\author{Sándor P. Fekete}{Department of Computer Science, TU Braunschweig}{fekete@tu-braunschweig.de}{https://orcid.org/0000-0002-9062-4241}{}
\author{Phillip Keldenich}{Department of Computer Science, TU Braunschweig}{keldenich@ibr.cs.tu-bs.de}{https://orcid.org/0000-0002-6677-5090}{}
\author{Michael Perk}{Department of Computer Science, TU Braunschweig}{fekete@tu-braunschweig.de}{https://orcid.org/0000-0002-0141-8594}{}
\authorrunning{S.\,P.\ Fekete and P.\ Keldenich and M.\ Perk} 

\Copyright{Sándor P. Fekete and Phillip Keldenich and Michael Perk} 


\category{} 

\funding{
  The work presented in this paper was largely funded by the DFG grant 
  ``Computational Geometry: Solving Hard Optimization Problems'' (CG:SHOP), grant FE407/21-1.} 

\acknowledgements{}

\EventEditors{Oswin Aichholzer and Haitao Wang}
\EventNoEds{2}
\EventLongTitle{41st International Symposium on Computational Geometry (SoCG 2025)}
\EventShortTitle{SoCG 2025}
\EventAcronym{SoCG}
\EventYear{2025}
\EventDate{June 23--27, 2025}
\EventLocation{Kanazawa, Japan}
\EventLogo{socg-logo.pdf}
\SeriesVolume{332}
\ArticleNo{174}     

\newcommand{\lbEpsilon}{2.730751 \cdot 10^{-16}}
\newcommand{\lbDelta}{6.458762 \cdot 10^{-16}}
\newcommand{\lbN}{84}
\newcommand{\lbRho}{1.44116645381}
\newcommand{\lbRhoShort}{1.44116}
\begin{document}

\supplement{Code, experiment instances and results are archived on Zenodo.}
\supplementdetails[subcategory={Source Code}]{Software}{https://doi.org/10.5281/zenodo.14266122}
\supplementdetails[subcategory={Experiment Data}]{Dataset}{https://doi.org/10.5281/zenodo.14266122}

\maketitle

\begin{abstract}
We provide a spectrum of new theoretical insights and practical results 
for finding 
a Minimum Dilation Triangulation (MDT), a natural geometric optimization 
problem of considerable previous attention:
Given a set $P$ of $n$ points in the plane, find a triangulation
$T$, such that a shortest Euclidean path in $T$ between any pair of points
increases by the smallest possible factor compared to their
straight-line distance. No polynomial-time algorithm is known for the problem;
moreover, evaluating the objective function involves computing the sum
of (possibly many) square roots. 
On the other hand, the problem is not known to be \NP-hard.

(1) We provide practically robust methods and implementations for computing an MDT
for benchmark sets with up to 30,000 points in reasonable time on commodity
hardware, based on new geometric insights into the structure of optimal edge
sets. Previous methods only achieved results for up to $200$ points, so we extend 
the range of optimally solvable instances by a factor of $150$.

(2) We develop scalable techniques for accurately
evaluating many shortest-path queries that arise as large-scale sums of square
roots, allowing us to certify exact optimal solutions,  
with previous work relying on (possibly inaccurate) floating-point computations.

(3) We resolve an open problem by establishing a lower bound of
$\lbRhoShort$ on the dilation of the regular $\lbN$-gon (and thus for arbitrary point
sets), improving the previous worst-case lower bound of $1.4308$ 
and greatly reducing the remaining gap to the upper bound of
$1.4482$ from the literature. In the process, we provide optimal solutions for regular
$n$-gons up to $n = 100$.
\end{abstract}

\input{content/01_introduction.tex}
\input{content/02_preliminaries.tex}
\input{content/03_possible_edges.tex}
\input{content/04_algorithm.tex}
\input{content/05_experiments.tex}
\input{content/06_conclusion.tex}



\bibliography{main.bib}

\appendix

\input{content/07_appendix.tex}

\end{document}

%% file: content/01_introduction.tex
\section{Introduction}
\label{sec:introduction}
Triangulating a set of points is one of the classical problems in computational
geometry. On the practical side, it has natural applications in wireless sensor
networks~\cite{DBLP:journals/comcom/WuLC07,DBLP:conf/infocom/ZhouWXJD11}, mesh
generation~\cite{bern1995mesh}, computer
vision~\cite{DBLP:conf/evoW/Vite-SilvaCTF07}, geographic information
systems~\cite{DBLP:journals/gis/Tsai93} and many other
areas~\cite{DBLP:books/lib/BergCKO08}. On the theoretical side,
finding a triangulation that is optimal with respect to some 
objective function has also received considerable attention:
The Delaunay triangulation maximizes the minimum angle and minimizes the
maximum circumcircle of all triangles.  Minimizing the
maximum edge length is possible in quadratic
time~\cite{DBLP:journals/siamcomp/EdelsbrunnerT93}.  On the other hand,
maximizing the minimum edge length is
\NP-complete~\cite{DBLP:journals/jocg/FeketeHHST18}. Famously,
Mulzer and Rote~\cite{DBLP:journals/jacm/MulzerR08} showed that 
computing the Minimum Weight Triangulation (MWT) 
is \NP-hard.

In this paper, we provide new results and insights for another natural
objective that considers triangulations as sparse structures with relative low
cost for ensuing detours: The \emph{dilation} of a triangulation $T$ of a point
set $P$ is the worst-case ratio (among all $s, t \in P$) between the shortest
$s$-$t$-path in $T$ and the Euclidean distance between $s$ and $t$.  The
Minimum Dilation Triangulation (MDT) problem asks for a triangulation that
minimizes the dilation~$\rho$, see~\cref{fig:example-instances} for examples.  This problem is closely related to the concept of a
Euclidean $t$-spanner: a subgraph 
with dilation (also called \emph{spanning ratio} or \emph{stretch
factor}~\cite{DBLP:journals/siamcomp/NarasimhanS00}) at most~$t$.  
Spanners have application in areas as robotics, network
design~\cite{DBLP:journals/tpds/AlzoubiLWWF03,DBLP:journals/dam/FarleyPZW04},
sensor networks~\cite{DBLP:journals/siamdm/CaiC95,DBLP:conf/sensys/FanLS06} and
design of parallel machines~\cite{DBLP:journals/comgeo/AronovBCGHSV08} and
have been studied extensively~\cite{DBLP:journals/ijcga/ChandraDNS95}.
Computing the MDT amounts to computing a \emph{plane} spanner with 
smallest spanning ratio, as every plane spanner can be extended into a
triangulation.  
For many types of triangulations, such as the Delaunay triangulation
or the MWT, both lower and upper bounds on the worst-case dilation are
known. Moreover, a constant upper bound on the worst-case dilation of a
triangulation also implies a constant-factor
approximation of the MDT.  Lower and upper bounds on the
worst-case dilation of the MDT have been studied, both for general
point sets and for special point sets such as regular polygons or points on a
circle; see \cref{subsec:related} for further details.

\begin{figure}
    \includegraphics[width=\linewidth]{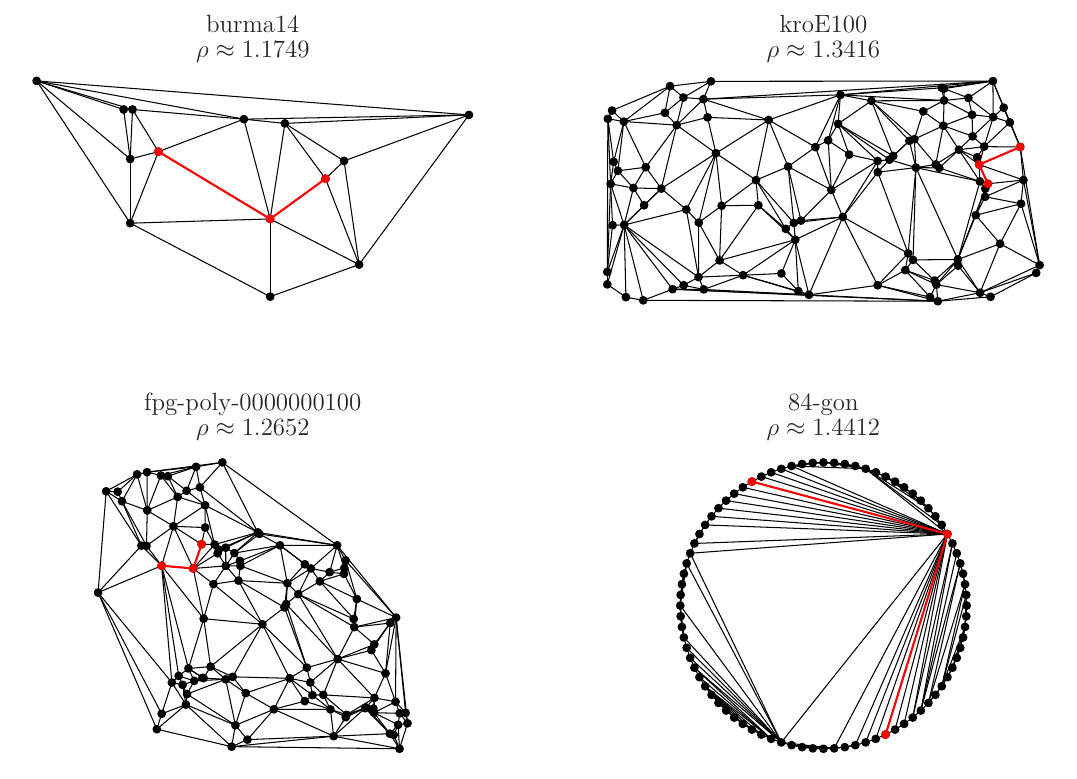}
    \caption{MDT solutions for four instances. The red edges indicate a dilation-defining path.}
    \label{fig:example-instances}
\end{figure}

Despite this importance and attention, actually computing a Minimum Dilation
Triangulation is a challenging problem. Its computational
complexity is still unresolved, signaling that there may not be a simple and
elegant algorithmic solution that scales well. Moreover, actually computing
accurate dilations involves computing shortest paths under Euclidean distances
between $\Theta(n^2)$ pairs of points requires dealing with another famous 
challenge~\cite{o1981advanced,TOPP,bloemer1991computing,eisenbrand2024improved},
as it corresponds to evaluating and comparing numerous sums of many square roots.

\subsection{Our contributions}
We provide various new contributions, both to theory and practice.

\begin{enumerate}
\item We present practically robust methods and implementations for computing
\chreplaced{an optimal}{the} MDT for benchmark sets with up to \num{30000} points in reasonable time on
commodity hardware, based on new geometric insights into the structure of optimal
edge sets. 
Previous methods only achieved results for up to \num{200} points
(involving one computational routine of complexity $\Theta(n^4)$ 
instead of our improved complexity of $O(n^2\log n)$), so we extend
the range of practically solvable instances by a factor of \num{150}. 
\item We develop scalable techniques for accurately evaluating many shortest-path queries 
that arise as large-scale sums of square roots, allowing us to certify exact optimal solutions.
This differs from previous work, which relied on floating-point computations, without regard for errors resulting from numerical issues.
\item We resolve an open problem from~\cite{DBLP:journals/ijcga/DumitrescuG16}
by establishing a lower bound of $\lbRhoShort$ on the dilation of the regular $\lbN$-gon
(and thus for arbitrary point sets).  
This improves the previous worst-case lower bound of $1.4308$ 
from the regular $23$-gon and greatly reduces the remaining gap to the upper
bound of $1.4482$ from~\cite{DBLP:journals/comgeo/SattariI19}.
In the process, we provide optimal solutions for regular $n$-gons up to $n=100$.
\end{enumerate}

\subsection{Related work}
\label{subsec:related}
The complexity of finding the MDT is unknown~\cite{DBLP:books/el/00/Eppstein00}.
\cite{DBLP:journals/ijcga/GiannopoulosKKKM10} prove that finding the minimum dilation graph with a limited number of edges is \NP-hard.
\cite{DBLP:journals/comgeo/CheongHL08} show that finding a spanning tree of given dilation is also \NP-hard.
Kozma~\cite{DBLP:conf/esa/Kozma12} proves \NP-hardness for minimizing the expected distance between random points in a triangulation,
with edge weights instead of Euclidean distances.
For surveys, see 
Eppstein~\cite{DBLP:books/el/00/Eppstein00} until 2000
and~\cite{DBLP:books/daglib/0017763} for more recent results.

On the practical side, 
the authors of~\cite{DBLP:journals/corr/abs-2305-11312} study the problem of finding sparse low dilation graphs for large point sets in the plane.
The authors of~\cite{DBLP:conf/esa/BuchinBGW24} present an approximation algorithm for the related problem of improving a given graph with a budget of $k$ edges such that the dilation is minimized.
Regarding the MDT, all practical approaches in the literature are based on fixed-precision arithmetic.
Klein~\cite{klein2006effiziente} used an enumeration algorithm to find an optimal MDT for up to \num{10} points. 
The authors of~\cite{DBLP:journals/heuristics/DorzanLMH14} present heuristics for the MDT and evaluate their performance on instances with up to \num{200} points.
Instances with up to \num{70} points were solved by 
the authors of~\cite{DBLP:conf/cccg/BrandtGSR14} using integer linear programming techniques.
In their approach, the edge elimination strategy from Knauer~and~Mulzer~\cite{DBLP:conf/ewcg/KnauerM05} was used to eliminate edges from the complete graph.
Recently, Sattari~and~Izadi~\cite{DBLP:journals/jgo/SattariI17} presented an exact algorithm based on branch and bound that was evaluated on instances with up to \num{200} points.

The MDT is closely related to finding a plane $t$-spanner; see~\cite{mitchell2017proximity} for an overview.
Chew~\cite{DBLP:conf/compgeom/Chew86} first proved an upper bound of $\sqrt{10}$ on plane $t$-spanners in the $L_1$-metric,
which he later improved~\cite{DBLP:journals/jcss/Chew89} to $2$ for the triangular-distance Delaunay graph in the plane.
\cite{DBLP:journals/jocg/BiniazAMSBC16} proved that any convex point set admits a plane $1.88$-spanner. 
For a centrally symmetric convex point set containing $n$ points, Sattari and Izadi~\cite{DBLP:journals/ipl/SattariI18} give an upper bound of $\nicefrac{n}{2} \sin(\nicefrac{\pi}{n})$.
The best known upper bound on the dilation of arbitrary point sets is by Xia~\cite{DBLP:journals/siamcomp/Xia13}, 
who established and upper bound of \num{1.998} for the dilation of the Delaunay triangulation.

Mulzer~\cite{mulzer2004minimum} studied the MDT for the set of vertices of a regular $n$-gon and proved an upper bound of $1.48586$. 
Amarnadh~and~Mitra~\cite{DBLP:conf/iccsa/AmarnadhM06} improved this bound to $1.48454$ for any point set that lies on the boundary of a circle.
Sattari and Izadi~\cite{DBLP:journals/comgeo/SattariI19} again improved the bound to $1.4482$.
Dumitrescu and Ghosh~\cite{DBLP:journals/ijcga/DumitrescuG16} show that any triangulation of a regular $23$-gon has dilation at least $1.4308$,
improving upon the bound of $1.4161$ by Mulzer~\cite{mulzer2004minimum} and answering a question posed by Bose and Smit~\cite{DBLP:journals/comgeo/BoseS13} as well as Kanj~\cite{DBLP:conf/iccit/Kanj13}.
Dumitrescu~and~Ghosh~\cite{DBLP:journals/ijcga/DumitrescuG16} also computed dilations of regular $22$-gon and regular $24$-gon.

%% file: content/02_preliminaries.tex
\section{Preliminaries}
\label{sec:preliminaries}

Let $P \subset \mathbb{R}^2$ be a set of points in the plane.
We denote the Euclidean distance between two points $u,v \in P$ by $d(u,v)$.
For a connected geometric graph $G = (P, E)$ with $E \subseteq {P \choose 2}$, we denote the Euclidean shortest path between two points $u,v \in P$ by $\pi_G(u,v)$ and its length by $|\pi_G(u,v)|$,
omitting $G$ if it is clear from context.
The dilation $\rho_G(u,v)$ between two points $u,v$ in $G$ is the ratio $\rho_G(u,v) := \frac{|\pi_G(u,v)|}{d(u,v)}$ between the shortest path length and the Euclidean distance.
The dilation $\rho(G)$ of the graph $G$ is defined as the maximum dilation between any two points in $P$,
i.e., $\rho(G) := \max \{ \rho_G(u,v) \mid u,v \in P, u \neq v\}.$

In the remainder of this work, the graph $G$ we consider is a triangulation, i.e., a maximal crossing-free graph on $P$.
Two edges $e_1 = (p_1, q_1), e_2 = (p_2, q_2)$ are said to \emph{cross} or \emph{intersect} iff the line segments they induce intersect in their interior.
Given a point set $P$, the Minimum Dilation Triangulation problem (MDT) asks to find a triangulation $T$ of $P$ minimizing $\rho(T)$.

%% file: content/03_possible_edges.tex
\section{Candidate edge enumeration}
\label{sec:edge-enumeration}

Here we describe a novel and practically efficient scheme for
enumerating a set of edges that induces a supergraph of the MDT.
We start with the underlying theoretical ideas for this supergraph, followed
by an algorithm for computing a supergraph of any triangulation with dilation \emph{strictly less} than a given bound $\rho$,
which we exploit to enumerate a supergraph of the MDT with a (usually) small number of edges.
This is further adapted to reductions of the dilation bound $\rho$.
We also discuss the computation of lower bounds on the dilation of the MDT.
We defer the discussion of some implementation details to \cref{sec:implementation-details}.

\subsection{Theoretical background}
Our supergraph is based on the well-known \emph{ellipse property} (used in \cite{DBLP:conf/cccg/BrandtGSR14,DBLP:journals/ijcga/GiannopoulosKKKM10,DBLP:conf/ewcg/KnauerM05}) that all edges of a triangulation $T$ with dilation below $\rho$ must satisfy.

\begin{definition}
  A pair of points $s, t$ has the \emph{ellipse property} with 
  respect to a point set $P$ and a dilation bound $\rho$ if, for
  any pair of points $\ell, r \in P \setminus \{s,t\}$ such that
  $\ell r$ and $st$ intersect, $\min \{d(\ell,s) + d(s,r), d(\ell,t) + d(t,r)\} < \rho d(\ell,r)$. 
\end{definition}

Recall that the set of points that have the same sum of distances $\kappa$ to 
two points $\ell, r$ with $\kappa \geq d(\ell,r)$ is an \emph{ellipse} with $\ell, r$ as its focal points (or \emph{foci});
thus, all paths between $\ell$ and $r$ with length less than $\kappa = \rho \cdot d(\ell,r)$ must lie strictly inside this ellipse.

If a pair of points $s,t$ does not have the ellipse property, then there is a pair of points $\ell, r$
such that $st$ cuts through all paths between $\ell$ and $r$ that could have dilation less than $\rho$; see \cref{fig:ellipse-property-example}.
We call such a pair of points $\ell, r$ an \emph{exclusion certificate} for $s,t$.
\begin{figure}%
  \centering
  \includegraphics{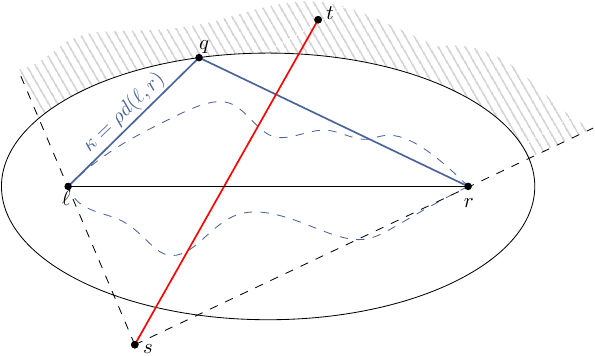}
  \caption{
    Any path connecting $\ell$ and $r$ with length below $\kappa = \rho d(\ell,r)$ must lie within the ellipse (dashed black lines).
    The edge $st$ does not have the \emph{ellipse property} as neither $s$ nor $t$ lie inside the ellipse;
    therefore, inserting $st$ makes connecting $\ell$ and $r$ by a sufficiently short path impossible.
  }
  \label{fig:ellipse-property-example}
\end{figure}%

\begin{observation}
  Let $\rho \geq 1$ be some dilation bound and let $T$ be a triangulation containing the edge $st$
  for points $s,t$ that do not have the ellipse property w.r.t.\ $P$ and $\rho$.
  Then, $\rho(T) \geq \rho$.
\end{observation}

\subsection{High-level description}
\label{sec:ellipse-filter-algo}
Preliminary experiments showed that a brute-force check of each of the $\Theta(n^4)$ pairs of
potential edges is impractical for large $n$, dominating the overall runtime time even for
early versions of our solution approach.

We therefore developed a more efficient scheme to enumerate a superset of the edges that satisfy the ellipse property.
This scheme performs a \emph{filtered incremental nearest-neighbor search} from each point $p \in P$
to identify all candidate edges of the form $pt$, i.e., looking for all possible \emph{neighbors} $t$ of $p$.
This search is efficiently implemented on a quadtree containing all points.
While enumerating candidates, we construct so-called \emph{dead sectors}, i.e., regions of the plane that
cannot contain possible neighbors of $p$.
We exclude all points that lie in dead sectors; we also use dead sectors to prune entire nodes of the quadtree
and to terminate the search early if it has become clear that all remaining points must be in a dead sector.
This usually avoids considering most points as potential neighbors of $p$ individually.
The algorithm has a runtime of $O(nk\log n)$, where $k$ is the average number of points and quadtree vertices considered individually from each point.
At worst, this can be $O(n^2\log n)$; in practice, $k$ is often much lower than $n$.
A related, slightly less complex enumeration algorithm applied to minimum-weight triangulations by Haas~\cite{haasmwt} scales to $10^8$ points.
In the following, we describe the components of the enumeration scheme in more detail.

\subsection{Dead sectors}
\label{sec:dead-sectors}
We begin by giving a definition of the dead sectors we use.
\begin{definition}
  Given a dilation $\rho$, a source point $p$ and two points $\ell, r$,
  the \emph{dead sector} $\mathcal{DS}_{\rho}(p, \ell, r) \subset \mathbb{R}^2$ is the region of all points $t$ such that 
  $pt$ intersects $\ell r$ and neither $p$ nor $t$ lie in the ellipse with foci $\ell, r$ and focal distance sum 
  $\kappa = \rho d(\ell, r)$.
\end{definition}
Depending on $p$, $\rho$, $\ell$ and $r$, $\mathcal{DS}_{\rho}(p, \ell, r)$ is either empty (if $p$ is in the ellipse)
or it is bounded by two rays and an elliptic arc; see \cref{fig:dead-sector-construction}.
In that case, it can also be interpreted as an elliptic arc and an interval of polar angles around $p$.

During our enumeration we construct many dead sectors,
the union of which can become quite complex, making it cumbersome and inefficient to work with directly.
We instead chose to simplify the shape of our dead sectors, giving up 
a small fraction of excluded area in exchange for a simple and efficient representation.
We replace the elliptic arc by a single disk centered on $p$,
whose radius is at least the maximum distance from $p$ to any point on the elliptic arc.
We can hence represent each non-empty simplified dead sector by a polar angle interval around $p$ 
and a single radius called \emph{activation distance} $A_{\rho}(p,\ell,r)$; see \cref{fig:dead-sector-construction}.

\begin{figure}
  \centering
  \includegraphics{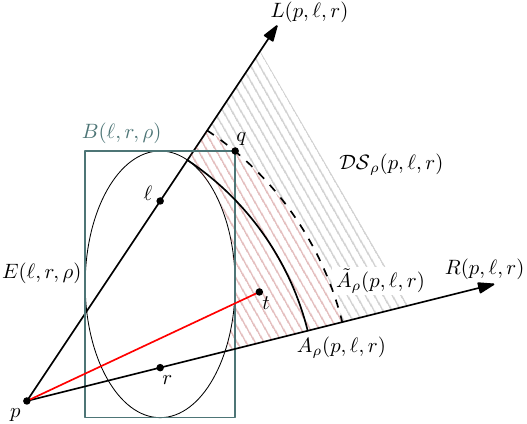}
  \caption{A dead sector $\mathcal{DS}_{\rho}(p, \ell, r)$, 
  shaded in gray and red with the ellipse $E(\ell, r, \rho)$.
  We approximate the ellipse by a disk centered at $p$ with radius $\tilde{A}_{\rho}(p,\ell,r)$ (thereby ignoring the red area),
  which is the distance from $p$ to the farthest point $q$ of the rectangle 
  $B(\ell, r, \rho)$ in $\mathcal{DS}_{\rho}(p, \ell, r)$.}
  \label{fig:dead-sector-construction}
\end{figure}

We initially attempted to compute fairly precise upper bounds on the approximation distance;
however, due to computational and numerical issues described in \cref{sec:precise-activation-distances},
we decided to use a more robust and efficient approach.
Instead of using the elliptic arc,
we compute an upper bound $\tilde{A}_{\rho}(p,\ell,r)$ using a minimal rectangle $B(\ell, r, \rho)$ containing the ellipse with sides are parallel and perpendicular to $\ell r$.
We only need to check the extreme points of $B(\ell, r, \rho)$ and its intersections with the rays $L(p,\ell,r)$ and $R(p,\ell,r)$ to compute an upper bound $\tilde{A}_{\rho}(p, \ell, r)$; see \cref{fig:dead-sector-construction}.

In the worst case, these simplifications may lead to additional candidate edges. 
While this cannot make the resulting graph exclude any edges that satisfy the ellipse property,
it is still undesirable; we use additional checks for further reductions later on.

\subsection{Quadtree}
We use a quadtree containing all points in $P$ for the filtered incremental search.
The points are stored in a contiguous array $A_P$ outside the tree.
Each quadtree node $v$ is associated with a contiguous subrange of $A_P$ represented by two pointers.
This subrange contains the points in the subtree $\mathcal{T}_v$ rooted at $v$.
Each node also has a bounding box that contains all points in $\mathcal{T}_v$.
Each interior node has precisely four children; each leaf node contains at most a small constant number of points.
To allow the contiguity of the subranges, points are reordered during tree construction.
This has the added benefit of spatially sorting the points,
improving the probability that geometrically close points are near each other in memory~\cite{haasmwt}.

\subsection{Enumeration process}
One can think of the filtered incremental search from $p$ as a process of continuously growing a disk centered at $p$ starting with a radius of $0$.
As in the sweep-line paradigm, one encounters different types of events at discrete disk radii.
We primarily encounter events when the disk first touches a point of $P$ or the bounding box of a quadtree node.

Observe that, during a search from a point $p$, the dead sectors have two states:
either their activation distance is not yet reached, in which case they are \emph{inactive} and do not exclude any points,
or they are \emph{active} and exclude all points in a certain polar angle interval around $p$.
We therefore also introduce events when the disk radius reaches the activation distance of a dead sector.
This enables efficient management of active dead sectors as a set of polar angle intervals around $p$;
we discuss ensuing numerical issues in \cref{sec:exactness-implementation-issues}.

When we first encounter a point $t$, we have to determine whether $t$ is in any dead sector by checking the active dead sector data structure.
If it is not, we have to report it as potential neighbor of $p$, adding it to a set of points sorted by polar angle around $p$.
Furthermore, to construct new dead sectors, we combine $t$ with $O(1)$ other points of $P$;
which points we use is decided by a heuristic discussed in \cref{sec:dead-sector-construction-neighbors}.

When we encounter a quadtree node $v$, we have to determine whether $v$'s bounding box is fully contained 
in the union of all dead sectors and can thus be pruned; we thus again check the active dead sectors.
Otherwise, we have to take $v$'s children, or the points it contains if it is a leaf, into account;
they are then considered as future events.

\subsection{Initialization and postprocessing}
We initially compute the Delaunay triangulation of the point set $P$ and compute its dilation.
We also optionally attempt to improve the dilation of the triangulation by a simple improvement heuristic.
The heuristic is based on computing constrained Delaunay triangulations,
greedily adding shortcut edges as constraints to reduce the length of the path currently defining the dilation.
We then use the resulting dilation $\rho$ as bound for the enumeration process outlined in the previous sections,
enumerating only edges that could locally be in a triangulation with dilation strictly below $\rho$.

After the initial enumeration process is complete, we are left with a set of \emph{possible} edges and can safely ignore all other edges.
We postprocess these as follows.
For each possible edge $pq$, we compute the set of possible edges intersecting it.
We need this information later on to model the problem of finding a triangulation on the set of possible edges.
For each pair $st$, $\ell r$ of intersecting edges, we explicitly check whether either pair is an exclusion certificate for the other.
In many cases, this postprocessing gets us very close to the edge set that would be obtained by the trivial $\Theta(n^4)$ edge candidate enumeration algorithm;
see the experiment section for details.
We mark each edge that does not have intersecting possible edges as \emph{certain};
certain edges must be part of any triangulation with dilation less than $\rho$.

\subsection{Dilation thresholds}
We also compute a \emph{dilation threshold} $\vartheta(st)$ for each possible edge $st$.
Let $I(st)$ be the set of possible edges intersecting $st$.
For each $\ell r \in I(st)$, we can compute a lower bound on the dilation of the 
shortest path connecting $\ell$ to $r$, should $st$ be present, as \[\rho_{st}(\ell r) = \min \{d(\ell,s) + d(s,r), d(\ell,t) + d(t,r)\}/d(\ell,r);\]
see \cref{fig:dilation-thresholds}.
\begin{figure}
  \centering
  \includegraphics{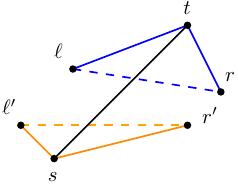}
  \caption{If $st$ (black) is present, the pairs $\ell, r$ (resp.\ $\ell',r'$)
           can only be connected by paths that have at least the length of the solid blue (resp.\ orange) path.
           Hence, the ratio between the solid and dashed blue (resp. orange) paths is a lower bound on the dilation of any triangulation containing $st$.
           The maximum of such lower bounds for $st$ is the dilation threshold $\vartheta(st)$ of $st$.}
  \label{fig:dilation-thresholds}
\end{figure}
The dilation threshold of $st$ is then $\vartheta(st) = \max_{\ell r \in I(st)} \rho_{st}(\ell r)$.
\begin{observation}
  If an edge $st$ has dilation threshold $\vartheta(st)$, it is not
  in any triangulation with dilation $\rho < \vartheta(st)$.
\end{observation}
This allows us to quickly exclude further edges if we lower the dilation bound $\rho$ we aim for,
without reenumerating edges or recomputing intersections.

\subsection{Lower bounds}
We use the dilation thresholds to bound the minimum dilation.
For each edge $st$, either $st$ or some edge crossing it must be in any triangulation.
Thus, the minimum of all dilation thresholds among $\{st\} \cup I(st)$ is a lower bound on the minimum dilation.
Combining all edges, we obtain the following lower bound:
\[\rho(T) \geq \max_{st} \min \{\vartheta(pq) \mid pq \in \{st\} \cup I(st) \}.\]
We use interval arithmetic to compute a safe lower bound on this value in $O(1)$ time per intersection between two edges resulting from our enumeration.

Furthermore, by computing the points on the convex hull, we also know the number of edges $g$ of any triangulation.
This also gives us a lower bound on the minimum dilation by considering the $g$ lowest dilation thresholds. 
We also use Kruskal's algorithm to compute the lowest dilation threshold that admits a connected graph.

%% file: content/04_algorithm.tex
\section{Exact algorithms}
\newcommand{\binmdt}{\textsc{BinMDT}}%
\newcommand{\incmdt}{\textsc{IncMDT}}%
\label{sec:exact-algorithms}%
Now we present two exact algorithms:
\incmdt{} is an incremental method that uses a SAT solver for
iterative improvement, until it can prove that no better solution exists.
\binmdt{} is based on a binary search for the optimal dilation $\rho$;
once the lower and upper bound are reasonably close,
the approach falls back to \incmdt{} to reach a provably optimal solution.

\subsection{Triangulation supergraph}
Both algorithms rely on the MDT supergraph mentioned in \cref{sec:edge-enumeration}.
As part of this computation, we also obtain an initial triangulation and its dilation,
as well the intersecting possible edges $I(st)$ for each possible edge $st$.
In both algorithms, we may gradually discover triangulations with lower dilation;
these are used to exclude additional edges using the precomputed dilation thresholds $\vartheta(e)$.
To keep track of the status of each edge,
we insert all points and possible edges into a graph data structure we call \emph{triangulation supergraph}.
In this structure, we mark each edge as \emph{possible}, \emph{impossible} or \emph{certain}.
Initially, all enumerated edges are \emph{possible}.
If, at any point, all edges intersecting an edge $e$ become \emph{impossible}, $e$ becomes \emph{certain}.
If an \emph{impossible} edge becomes \emph{certain} or vice versa, the graph does not contain a triangulation any longer.
If this happens, we say we encounter an \emph{edge conflict}.

\subsection{SAT formulation}
Given a triangulation supergraph $G = (P, E)$, we model the problem of finding a triangulation 
on \emph{possible} and \emph{certain} edges using the following simple SAT formulation.
Let $E_p \subseteq E$ be the set of non-\emph{impossible} edges when the SAT formulation is constructed.
For each edge $e \in E_p$, we have a variable $x_e$.
We use the following clauses in our formulation.
\begin{align}
    &\lnot x_{e_1} \lor \lnot x_{e_2} &\forall e_1, e_2 \in E_p: e_2 \in I(e_1) \label{eq:pairwise-intersection}\\
    &x_e \lor \bigvee_{\substack{e_j \in I(e)}} x_{e_j} &\forall e \in E_p\label{eq:enforce-edges}
\end{align}
Clauses (\ref{eq:pairwise-intersection}) ensure crossing-freeness and clauses (\ref{eq:enforce-edges}) ensure maximality.
When an edge $e$ becomes certain, we add the clause $x_e$; similarly, when an edge becomes impossible, we add the clause $\lnot x_e$.
Both algorithms are based on this simple formulation;
in the following, we describe how they use and modify it to find an MDT.

\subsection{Clause generation}
The following subproblem, which we call \emph{dilation path separation}, arises in both our algorithms:
Given a dilation bound $\rho$, a triangulation supergraph $G = (P, E)$ excluding only edges that cannot be in any triangulation with dilation less than $\rho$,
a current triangulation $T$ and a pair of points $s,t \in P$ such that $|\pi_T(s,t)| \geq \rho \cdot d(s,t)$, 
find a clause $C$ that is (a) violated by $T$ and (b) satisfied by any triangulation $T'$ with $\rho(T') < \rho$.

\begin{lemma}
    Assuming a polynomial-time oracle for comparing sums of square roots,
    there is a polynomial-time algorithm that solves the dilation path separation problem.
\end{lemma}
\begin{proof}
    Let $\Pi$ be the set of all $s$-$t$-paths $\pi$ in $G$ with $|\pi| < \rho \cdot d(s,t)$.
    We begin by observing that, along every path $\pi \in \Pi$, there is an edge $e \in E$
    that is not in $T$; otherwise, we get a contradiction to $|\pi_T(s,t)| \geq \rho \cdot d(s,t)$.
    Let $E' \subseteq E \setminus T$ be a set of edges such that for each $\pi \in \Pi$, 
    there is an edge $e \in E'$ on $\pi$.
    Then, $C = \bigvee_{e \in E'} x_e$ is a clause that satisfies the requirements;
    note that if $\Pi$ is empty, the empty clause can be returned.

    $T$ contains no edge from $E'$, so $C$ is violated by $T$.
    Furthermore, if a triangulation $T'$ with $\rho(T') < \rho$ does not contain any of the edges in $E'$,
    it contains none of the paths in $\Pi$.
    Therefore, $\pi_{T'}(s,t)$ uses an edge that is not in $E$, which has been excluded from all triangulations with dilation less than $\rho$; a contradiction.
    $E'$ can be computed by repeatedly computing shortest $s$-$t$-paths $\pi$;
    as long as $\pi < \rho d(s,t)$, we find an edge $e \notin T$ on $\pi$, add $e$ to $E'$ and forbid it in future paths.
    The number of edges bounds the number of iterations of this process;
    using the comparison oracle, we can efficiently perform each iteration.
\end{proof}

For a description of how we compute $E'$ in practice, see \cref{sec:practical-dilation-path-sep}.

\begin{figure}
    \includegraphics[width=\linewidth]{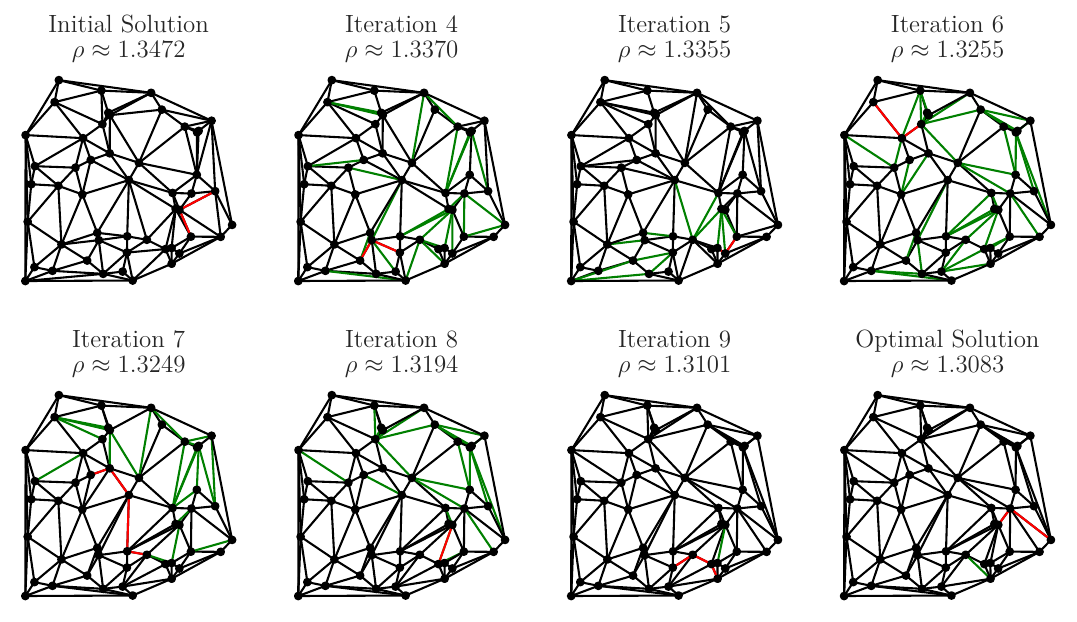}
    \caption{Progress of the incremental algorithm on an instance with $n = 50$ points. Green edges indicate changes in the triangulation, red edges indicate a dilation-defining path.}
    \label{fig:iterative-search-progress}
\end{figure}

\subsection{Incremental algorithm}
\label{sec:incremental-algorithm}
Based on the SAT formulation and the algorithm for the dilation path separation problem, \incmdt{} is simple.
Given an initial triangulation $T$ with dilation $\rho$, we enumerate the set of candidate edges and construct a triangulation supergraph $G$ with bound $\rho$.
We construct the initial SAT formula $M$ and solve it; if it is unsatisfiable, the initial triangulation is optimal.
Otherwise, we repeat the following until the model becomes unsatisfiable or we encounter an edge conflict, keeping track of the best triangulation found, see~\cref{fig:iterative-search-progress}.

We extract the new triangulation $T'$ from the SAT solver and compute the dilation $\rho'$ and a pair $s, t$ of points realizing $\rho'$.
If $\rho'$ is better than the best previously found dilation $\rho$, we update $\rho$ and mark all edges $e$ with $\vartheta(e) \geq \rho'$ as \emph{impossible}.
We then set $T = T'$ and solve the dilation path separation problem for $\rho$, $G$, $T$, $s$ and $t$.
We add the resulting clause to $M$ and let the SAT solver find a new solution.

\subsection{Binary search}
\label{sec:binary-search}
\newcommand{\rhoLB}{\rho_{\text{lb}}}
\newcommand{\rhoUB}{\rho_{\text{ub}}}

%
%
Preliminary experiments with \incmdt{} showed that we spend almost all runtime for computing dilations, 
even for instances for which we could rely exclusively on interval arithmetic, requiring no exact computations.
For many instances, most iterations of \incmdt{} resulted in tiny improvements of the dilation.
To reduce the number of iterations (and thus, dilations computed), 
we considered the binary search-based algorithm \binmdt{}.

\subsubsection{High-level idea}
At any point in time, aside from the dilation $\rhoUB$ of the best known triangulation,
\binmdt{} maintains a lower bound $\rhoLB$ on the dilation, initialized as described in \cref{sec:edge-enumeration}.

As long as $\rhoUB-\rhoLB \geq \sigma$ for a small threshold value $\sigma$, \binmdt{} performs a binary search.
It computes a new dilation bound $\rho = \frac{1}{2}(\rhoLB + \rhoUB)$.
It then uses the SAT model in a similar way as \incmdt{} to determine whether a triangulation $T$ with $\rho(T) < \rho$ exists.
If it does, it updates $\rhoUB = \rho(T)$; otherwise, it updates $\rhoLB = \rho$.

Once $\rhoUB-\rhoLB$ falls below $\sigma$, \binmdt{} falls back to a slightly modified version of \incmdt{} to find the MDT,
starting from the best known triangulation with dilation $\rhoUB$.

In the following, we describe and motivate the differences between how \incmdt{} and \binmdt{} use the SAT formulation;
for more details, see also \cref{sec:incremental-sat-solving}.

\subsubsection{Dilation sampling}
To further reduce the time spent on computing dilations, observe the following.
When a node $v$ of some graph $G$ is expanded in Dijkstra's algorithm from source $s$, we know the shortest path from $s$ to $v$ and thus the dilation $\rho_G(s,v)$.
Because $\rho_G(s,v) \leq \rho(G)$, we can compute a lower bound on the dilation much faster than the precise value by only performing a constant number of node expansions from each point $p \in P$.
We call this \emph{sampling} of the dilation.
Given a bound $\rho$ on the dilation, we can sample a triangulation $T$ for violations, i.e., pairs $s,t$ of points with $\rho_T(s,t) \geq \rho$.
We observed that a dilation-defining path usually consisted of few edges;
thus, we have a good chance of finding it by sampling.

If it is likely that a new-found triangulation $T'$ violates a given bound $\rho$,
we can thus expect to save time by sampling for violations instead of computing the dilation exactly.
Sampling also allows us to use multiple violations to construct multiple clauses in each iteration,
potentially further reducing the number of iterations.
\binmdt{} uses sampling after each SAT call with a small constant limit on the number of violations.
If violations are found, no full dilation computation is required and violations are used to construct clauses.
Only if no violations are found, we compute the exact dilation;
ideally, this only happens once for each upper bound reduction in the binary search,
namely once we find a triangulation satisfying the current bound.
We also sample in the final improvement phase of \binmdt{}.
For an experimental overview on the number of times sampling was sufficient in comparison to the 
number of times the dilation had to be computed exactly, see \cref{sec:experiments-dilation-computation}.

%% file: content/05_experiments.tex
\section{Empirical evaluation}
\label{sec:experiments}

Now we present experiments to evaluate our algorithms.
We used Python 3.12, with a core module written in C++20 for all computationally heavy tasks;
the code was compiled with GCC 13.2.0 in release mode.
We use CGAL 5.6.1 for geometric primitives and exact number types, Boost 1.83 for utility functions and pybind11 2.12 for Python bindings
and use the incremental SAT solver CaDiCaL 1.9.5 via the PySAT interface for solving the SAT models.
All experiments were performed on Linux workstations equipped with AMD Ryzen 9 7900 CPUs with 12 cores/24 threads and \qty{96}{GiB} of DDR5-5600 RAM running Ubuntu 24.04.1.

\subsection{Research questions}
Our experimental evaluation aims to answer the following questions.
\begin{enumerate}
  \item[Q1] How does the edge enumeration algorithm from \cref{sec:edge-enumeration} compare to the brute force enumeration in terms of runtime and the number of edges eliminated? 
            Can we solely rely on this algorithm or should we reconsider using the brute force enumeration?
  \item[Q2] What is the quality of the initial solutions computed by the Delaunay triangulation and the constrained Delaunay triangulations from \cref{sec:edge-enumeration} compared to the optimal solution?
  \item[Q3] How do our approaches compare to existing algorithms for the MDT w.r.t.\ runtime and solution quality?
            Can we solve instances with thousands of points to provable optimality?
  \item[Q4] How do \binmdt{} and \incmdt{} compare?
            Which should be used for large instances?
\end{enumerate}

\begin{figure}
    \begin{subfigure}[c]{0.49\linewidth}
        \centering
        \includegraphics[width=\linewidth]{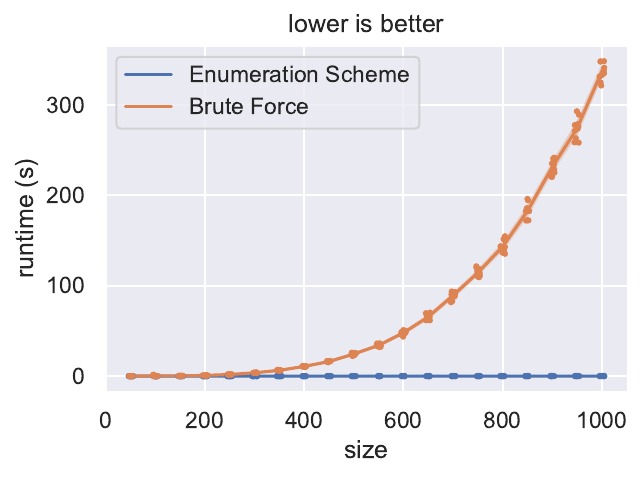}
    \end{subfigure}\hfill
    \begin{subfigure}[c]{0.49\linewidth}
        \centering
        \includegraphics[width=\linewidth]{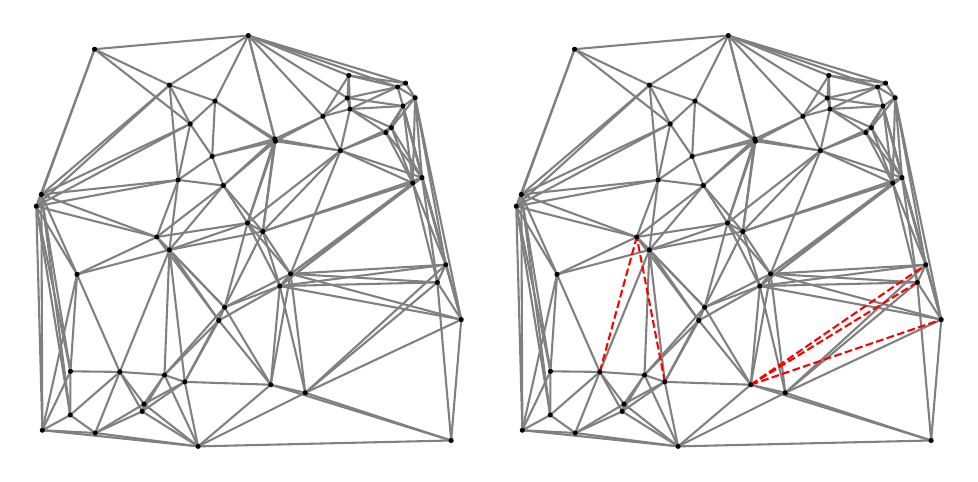}
    \end{subfigure}
    \caption{Comparison of our approach from \cref{sec:edge-enumeration} and the $\Theta(n^4)$ brute force elimination on the \emph{random} instance set.
             \textbf{(Left)} The $\Theta(n^4)$ elimination is infeasible for larger instances.
             \textbf{(Right)}~In all cases, only a small number of (red) edges were not eliminated by our approach.}
    \label{fig:edge-enumeration}
    \label{fig:edge-enumeration-edges}
\end{figure}

\subsection{Experiment design}
To answer our questions, we collected and generated a large set of instances, consisting of instances from the following instance classes.
In all cases, the coordinates of points in the instances are either integers or double precision floating-point numbers.
\begin{description}
  \item[random-small] 
    We include instances from the work of \cite{DBLP:conf/cccg/BrandtGSR14}.
    The $210$ instances have fixed sizes $n\in \{10,20,\dots,70\}$ and were generated by placing uniformly random points inside a $10 \times 10$ square.
    For each size a total of $30$ instances were generated.
  \item[random] 
    We include two sets of randomly generated instances.
    These encompass a set of instances with points with float coordinates chosen uniformly between $0$ and $10^3$,
    ranging from \num{50} to \num{10000} points.
    This resulted in a total of \num{800} instances.
  \item[public] 
    We include instances from all well-known publicly available benchmark instance sets we could locate.
    These include point sets previously used in the CG:SHOP challenges~\cite{demaine2022area,demaine2020computing},
    TSPLIB instances~\cite{reinelt1991}, instances from a VLSI dataset\footnote{https://www.math.uwaterloo.ca/tsp/vlsi/index.html} and
    point sets from the Salzburg Database of Polygonal Inputs~\cite{EDER2020105984}.
    In total, we collected \num{486} instances with up to \num{10000} points and an additional \num{38} with up to \num{30000} points.
\end{description}

\subsection{Q1: Edge enumeration}

We first compare the edge enumeration algorithm from \cref{sec:edge-enumeration} to the $\Theta(n^4)$ brute force enumeration of all possible edges on the \emph{random} instance set.
Both preprocessing options include finding all pairwise intersections between the enumerated segments.
Due to the $\Theta(n^4)$ runtime, we only consider instances with up to \num{1000} points; see~\cref{fig:edge-enumeration}.

The $\Theta(n^4)$ algorithm precisely identifies all edges that have the ellipse property;
it can hence only eliminate more edges than the edge enumeration algorithm from \cref{sec:edge-enumeration}.
However, for our test instances, the number of edges that can be eliminated by our approach is almost identical to the $\Theta(n^4)$ algorithm;
see \cref{fig:edge-enumeration-edges} for an example.
The runtime makes the $\Theta(n^4)$ algorithm infeasible for larger instances; 
it takes more than \qty{250}{s} for instances with only \num{1000} points,
compared to $<\qty{1}{s}$ for our more efficient approach.

\subsection{Q2: Initial solutions}
\label{sec:experiments-initial-solutions}

\begin{figure}
    \begin{subfigure}[t]{0.49\linewidth}
        \centering
        \includegraphics[width=\linewidth]{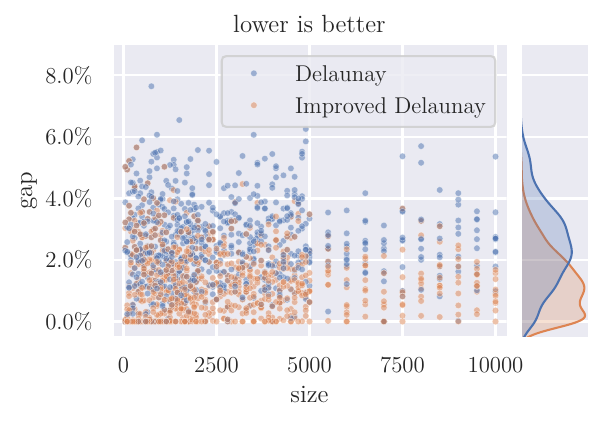}        
    \end{subfigure}\hfill
    \begin{subfigure}[t]{0.49\linewidth}
        \centering
        \includegraphics[width=\linewidth]{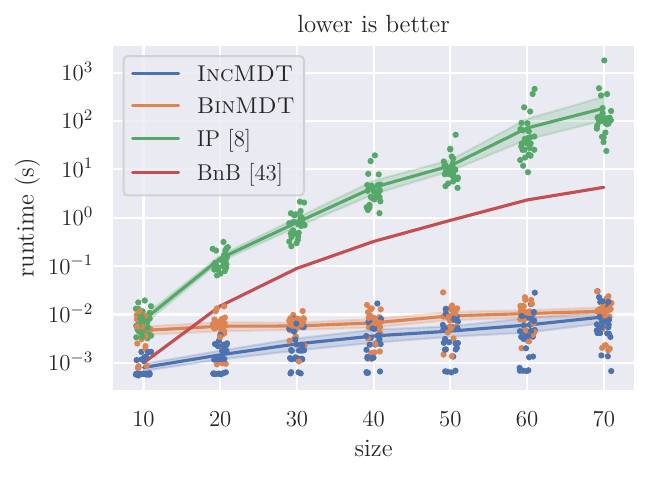}  
    \end{subfigure}
    \caption{\textbf{(Left)} Initial solution comparison on the \emph{random} set with up to \num{10000} points.
    Delaunay triangulations can be improved by shortcut edges, but are often close to the optimal solution.
    \textbf{(Right)}~Runtime comparison with the approaches from \cite{DBLP:conf/cccg/BrandtGSR14} and \cite{DBLP:journals/jgo/SattariI17} on the \emph{random-small} set.}
    \label{fig:initial-solutions-quality}
    \label{fig:brandt-runtime}
\end{figure}

Better initial solutions can reduce the number of candidate edges further than bad ones and thus affect the runtime of the overall algorithm.
Delaunay triangulations are a natural choice for the initial solution, but we suspect that they can be improved by adding shortcut edges.
\Cref{fig:initial-solutions-quality} shows that the relative dilation gap between the Delaunay triangulation and the optimal solution for the \emph{random} instance set.
It can be seen that the introduction of shortcut edges can indeed reduce the gap to the MDT to around \qty{1.5}{\%}.

\subsection{Q3: Comparison to state-of-the-art}
\label{sec:experiments-comparison-to-existing}

We compare our approaches to two exact state-of-the-art algorithms for the MDT.
Note that both of these use floating-point arithmetic and are not guaranteed to find the optimal solution.
However, we can confirm that all previous solutions are within a small relative error of our optimal solution.
The first approach is an integer programming (IP) approach by \cite{DBLP:conf/cccg/BrandtGSR14} that used the commercial software CPLEX to solve the MDT. 
The second comparison is with the most recent exact algorithm (BnB) from Sattari~and~Izadi~\cite{DBLP:journals/jgo/SattariI17}.
For both approaches the source code is no longer available, so we cannot compare our results on the same hardware.
Also for BnB~\cite{DBLP:journals/jgo/SattariI17} the instance data and results are no longer available. 
For both approaches we therefore use the data the authors published in their papers. 
Note that the drastic runtime difference that we see in our experiments cannot be attributed to improved hardware alone.

For \emph{random-small}, both \incmdt{} and \binmdt{} outperform the IP and BnB approach by a large margin (up to four orders of magnitude), see~\cref{fig:brandt-runtime}.
All instances were solvable in less than \qty{0.1}{s}.
Additionally, Sattari~and~Izadi~\cite{DBLP:journals/jgo/SattariI17} provided results for TSPLIB~\cite{reinelt1991} instances (part of our \emph{public} instance set) with up to \num{200} points.
\Cref{tab:tsplib-comparison} shows that our approach is faster than the algorithm from Sattari and Izadi.
We can solve instances to provable optimality in less than \qty{1}{s} while their algorithm took up to \qty{1248}{s}.

\subsection{Q4: Algorithm comparison}
\label{sec:experiments-algorithm-comparison}
\begin{figure}
    \begin{subfigure}[t]{0.49\linewidth}
        \centering
        \includegraphics[width=\linewidth]{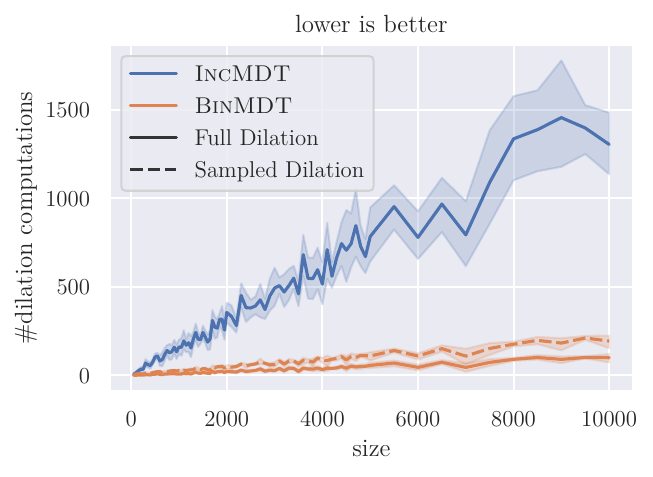}
    \end{subfigure}\hfill
    \begin{subfigure}[t]{0.49\linewidth}
        \centering
        \includegraphics[width=\linewidth]{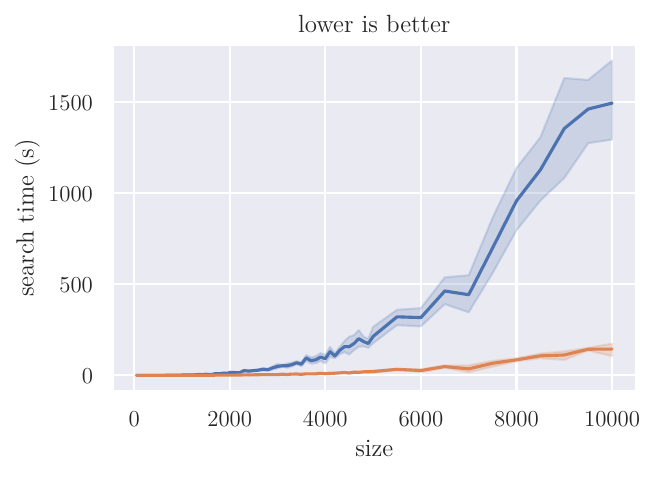}
    \end{subfigure}
    \caption{Experiments on the \emph{random} benchmark set.
             \textbf{(Left)} The number of dilation computations of \binmdt{} is significantly reduced compared to the incremental algorithm.
             \textbf{(Right)} \binmdt{} is significantly faster than the incremental algorithm.}
    \label{fig:experiments-algorithm-comparison}
\end{figure}
\begin{figure}
    \begin{subfigure}[t]{0.49\linewidth}
        \centering
        \includegraphics[width=\linewidth]{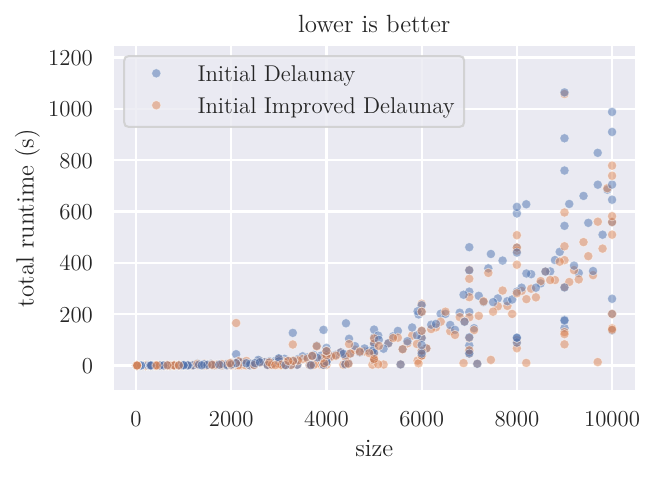}
    \end{subfigure}\hfill
    \begin{subfigure}[t]{0.49\linewidth}
        \centering
        \includegraphics[width=\linewidth]{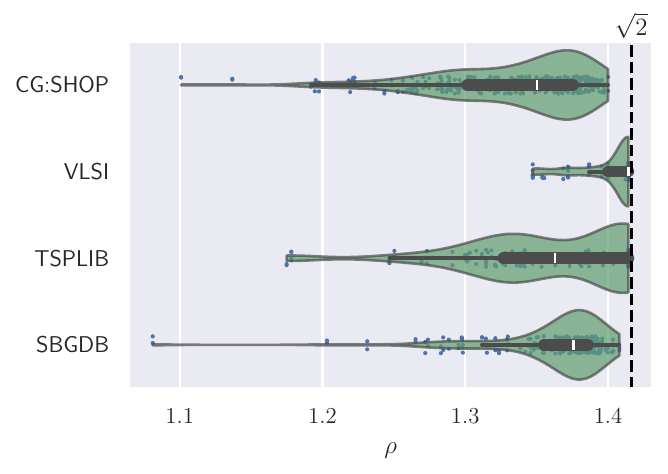}
    \end{subfigure}
    \caption{Experiments on the \emph{public} benchmark set.
             \textbf{(Left)} Using the improved Delaunay triangulation as an initial solution significantly improved the performance.
             \textbf{(Right)} The dilation of the MDTs is at most $\sqrt{2}$ for all instances.}
    \label{fig:experiments-public-instance-set}
\end{figure}

We now compare \incmdt{} to \binmdt{}; see~\cref{sec:appendix-algorithm-comparison} for more detail.
Recall that \binmdt{} aims to reduce the time spent on dilation computations by 
reducing the number of dilation computations and merely sampling the dilation whenever that suffices.
The first experiment was conducted on the \emph{random} instances with up to \num{10000} points;
this experiments confirms that \binmdt{} indeed achieves a significantly lower runtime,
requires fewer dilation computations, and that sampling is sufficient in most cases, see~\cref{fig:experiments-algorithm-comparison}.

After we confirmed that \binmdt{} is superior,
we conducted an additional experiment with two different initial solution strategies on the \emph{public} instances up to \num{10000} points, see~\cref{fig:experiments-public-instance-set}.
The first strategy uses the Delaunay triangulation as an initial solution; the second strategy uses the improved Delaunay triangulation with shortcut edges.
The improved Delaunay triangulation significantly reduces the runtime of \binmdt{} for almost all instances.
Detailed results for all instances of the \emph{public} instance set are shown in~\cref{tab:public-comparison}.
An additional experiment on the \emph{public} instance set shows that \binmdt{} can solve instances with up to \num{30000} points in less than \qty{17}{h} to provable optimality, see~\cref{tab:large-comparison}.

\section{Improved bounds for regular \texorpdfstring{$n$}{n}-gons}\label{sec:n-gon-lb}
Vertex sets of regular $n$-gons are particularly challenging, as there are no points in the
interior to serve as Steiner points for shortcuts, making 
local heuristics less successful. They are also natural candidates for worst-case
dilation, as each diagonal constitutes an obstacle for the separated points.
Thus, they have received considerable
attention~\cite{mulzer2004minimum,DBLP:journals/ijcga/DumitrescuG16,DBLP:journals/comgeo/SattariI19},
with a lower bound of $1.4308$ (see~\cite{DBLP:journals/ijcga/DumitrescuG16}) and an upper bound
of $1.4482$ (see~\cite{DBLP:journals/comgeo/SattariI19})
on their worst-case dilation. 
Improving this gap is an open question posed by~\cite[Problem 1]{DBLP:journals/ijcga/DumitrescuG16}, originating
from~\cite{DBLP:journals/comgeo/BoseS13, DBLP:conf/iccit/Kanj13}.
With the help of our exact algorithm (see~\cref{sec:appendix-n-gons}),
we were able to compute bounds for $n\in\{4,5,\dots,100\}$,
answering the problem by Dumitrescu and Ghosh.

Our implementation currently uses floating-point precision to represent the coordinates of the input points.
Thus, we cannot exactly represent the necessarily irrational points of any regular $n$-gon ($n \geq 5$) in our solver.
However, we can compute the MDT of a rational point set that is a good approximation of a regular $n$-gon.
We can then bound the error and thus obtain a rigorous lower bound on the dilation of the regular $n$-gon.
\begin{theorem}
    Let $P$ be a set of $n = \lbN$ points placed at the vertices of a regular $n$-gon.
    Then the dilation of the MDT of $P$ is $\rho \geq \lbRho$.
\end{theorem}
\begin{proof}
    Let $P$ be the point set of a regular $\lbN$-gon and $Q$ be a point set that contains floating-point approximations of the points of $P$.
    There is a bijection $f: P \to Q$ with inverse $f^{-1}$ that maps each point in $P$ to the closest point in $Q$.
    For a given triangulation or path of $P$, we write $f(\cdot)$ to denote the triangulation or path where the points are transformed by pointwise application of $f$.
    Neither $P$ nor $Q$ contain collinear points and all points are in convex position,
    therefore $T$ is a triangulation of $P$ iff $f(T)$ is a triangulation of $Q$.

    Let $\epsilon \geq \max_{p_i, p_j\in P}|d(p_i,p_j)-d(f(p_i),f(p_j))|$ be a bound on the maximum absolute error on distances and let
    $\delta$ be a chosen such that \[\forall p_i,p_j \in P: (1-\delta)d(p_i, p_j) \leq d(f(p_i), f(p_j)) \leq (1+\delta)d(p_i,p_j).\]

    Let $T$ be the MDT of $P$ with dilation $\rho$.
    Let $f(p_i), f(p_j)$ be a dilation-defining pair in $f(T)$ with path $\pi \subset Q$. 
    Because $T$ has dilation $\rho$, we have (i) $\nicefrac{|f^{-1}(\pi)|}{d(p_i,p_j)} \leq \rho$.
    As $\pi$ has at most $n-1$ edges, we have (ii) $|\pi| - (n-1)\epsilon \leq |f^{-1}(\pi)| \leq |\pi| + (n-1)\epsilon$ and thus the following upper bound on the MDT of $Q$.

    \begin{align*}
        \frac{|\pi|}{d(f(p_i),f(p_j))} = 
        \frac{|\pi| - (n-1)\epsilon}{d(f(p_i), f(p_j))} + \frac{(n-1)\epsilon}{d(f(p_i), f(p_j))} 
        &\overset{(ii)}{\leq}
        \frac{|f^{-1}(\pi)|}{(1-\delta)d(p_i, p_j)} + \frac{((n-1)\epsilon)}{d(f(p_i), f(p_j))} \\
        \overset{(i)}{\leq} \frac{1}{1-\delta}\rho + \frac{(n-1)\epsilon}{d(f(p_i), f(p_j))}
        &\leq \frac{1}{1-\delta}\rho + \frac{(n-1)\epsilon}{\min_{q_i,q_j\in Q}d(q_i,q_j)}.
    \end{align*}
    
    Using exact calculations done with \emph{sympy}, we computed $\epsilon \leq \lbEpsilon $ and $\delta \leq \lbDelta$,
    which, combined with the lower bound on the dilation of $Q$ from our solver gives us $\rho \geq \lbRho$.
\end{proof}

%% file: content/06_conclusion.tex
\section{Conclusion}
\label{sec:conclusion}
We have presented exact algorithms for minimum dilation triangulations,
greatly outperforming previous methods from the literature.
This has also yielded insights into
the intricate structure of optimal solutions for regular $n$-gons,
together with new lower bounds on the worst-case dilation of triangulations.
This demonstrates the value of computational tools for gaining analytic insights
that seem out of reach with purely manual analysis.

On the theoretical
side, one tantalizing open question remains the complexity of
the MDT; here the intricate structure of solutions for 
regular $n$-gons may show the way for hardness constructions even for 
convex arrangements. An equally fascinating problem is 
the worst-case dilation of triangulations; here we conjecture that 
the asymptotic dilation for regular $n$-gons 
corresponds to the worst-case dilation 
for general point sets, implying a value between
$\lbRhoShort$ and $1.4482$, with a better lower 
bound achievable via even larger $n$-gons.

On the practical side, the runtime of our algorithms 
is dominated by repeatedly computing the dilation of a triangulation.
With sums of square roots handled as described,
this does not appear to be a difficult problem per se,
so appropriately tailored algorithms or data structures 
for carrying out huge numbers of dilation queries
would be extremely helpful.

%% file: content/07_appendix.tex
\section{Implementation details}
\label{sec:implementation-details}
In this section, we describe relevant details omitted in the text above.

\subsection{Enumeration process}
\label{sec:dead-sector-construction-neighbors}
Here we describe the implementation details omitted from \cref{sec:edge-enumeration}.

\subsubsection{Precise activation distances}
\label{sec:precise-activation-distances}
Before resorting to the simplified activation distance computation described in \cref{sec:dead-sectors},
we attempted to compute precise upper bounds on the approximation distance using the elliptic arc and interval arithmetic as follows.
We begun by finding parameters $a,b,c,d,e,f$ that describe the ellipse using the equation $ax^2 + bxy + cy^2 + dx + ey + f = 0$,
based on $5$ points on the ellipse; this involves solving a linear system with $6$ variables and unknowns.

We then continued by a bisection search to find the maximum distance from $p$ to the full ellipse
on which the elliptic arc bounding the dead sector lies.
In each iteration of this procedure, we have to check whether a circle with center $p$ and some radius $r$ 
intersects the ellipse, and if the intersection is possibly between the two rays that form the other bounds of the dead sector.

The intersection points of this circle with the ellipse are the roots of a quartic equation.
Therefore, in each iteration, we have to check whether a quartic equation has a real root.
We found that this computation is costly and numerically unstable;
using interval arithmetic to compute a safe upper bound often produced overly conservative results.
This was mostly due to strong amplification of minute uncertainties contained in $a,\ldots,f$ when evaluating
the 6th-degree discriminant of the quartic equation to check for real roots;
even in cases that were not deliberately constructed to be numerically challenging,
we observed that the computation of the activation distance yielded results that were sometimes worse
than what the much cheaper bounding rectangle-based approach gave us.

\subsubsection{Dead sector construction}
We now discuss which other points we use to construct dead sectors when we encounter a point $t$.
During preliminary experiments, we found it insufficient to consider the closest points to $t$ w.r.t.\ the polar angle around $p$.
Points that are angularly near still can have high Euclidean distance, leading to a large activation distance or an empty dead sector.
For efficiency reasons, it may also be costly to consider large sets, such as all previously encountered points.
We only use a number that is $O(1)$ on average for each point $t$;
this is also needed to guarantee that the overall runtime of the algorithm is at most $O(n^2\log n)$:
otherwise, we might have significantly more dead sectors than points.

We use the following heuristics to determine \emph{interesting neighbors} of a point $t$ with which we attempt to construct dead sectors.
Firstly, we do consider a small constant number of previously encountered points $q$ that are close to $t$ in terms of polar angle.
Secondly, before running the actual enumeration process, we determine for each point a small number of interesting neighbors:
aside from the six closest neighbors of $t$, we also consider all neighbors of $t$ in the Delaunay triangulation.

\subsection{Exact algorithm}
In this section, we describe details omitted from \cref{sec:exact-algorithms}.

\subsubsection{Dilation path separation in practice}
\label{sec:practical-dilation-path-sep}
Recall that solving the dilation path separation problem is a crucial part of our exact algorithm.
We need to repeatedly and exactly compute shortest $s$-$t$-paths on our triangulation supergraph $G$
to build a set of edges $E' \subseteq E \setminus T$ without which there is no sufficiently short $s$-$t$-path;
we then return $\bigvee_{e \in E'} x_e$ as a clause that is violated by the current solution but satisfied by any solution with dilation less than $\rho$.

In order to speed up the computation, we begin by enumerating the vertices that may be on a sufficiently short $s$-$t$-path.
We start by running a bounded version of Dijkstra's algorithm from $s$ to enumerate the set $Q$ of all vertices $v$ that are reachable from $s$ in distance less than $\rho \cdot d(s,t) - d(v,t)$.
If this fails to reach $t$, we know that no sufficiently short path exists.
In this case, we can immediately return the empty clause, which is always violated, as no solution with dilation less than $\rho$ exists.

Otherwise, we further restrict $Q$ by running a second bounded Dijkstra's algorithm from $t$,
this time only considering vertices $v$ reachable from $t$ in distance less than $\rho \cdot d(s,t) - |\pi_G(s,v)|$,
using the shortest path distances computed by the run from $s$.

We also restrict the set of possible edges by checking for each edge between two possible vertices $u, v \in Q$ 
whether either of the two paths $\pi_G(s,v),\pi_G(u,t)$ or $\pi_G(s,u),\pi_G(v,t)$ is shorter than $\rho \cdot d(s,t)$.

Afterwards, we can repeatedly run a bidirectional version of Dijkstra's algorithm,
which simultaneously runs the algorithm from $s$ and $t$, restricted to the possible vertices and edges to find a shortest path $\pi$ between $s$ and $t$;
after each run, if the path is shorter than $\rho \cdot d(s,t)$, we add an edge $e \notin T$ from $\pi$ to $E'$,
ignoring $e$ in all subsequent bidirectional Dijkstra runs; see \cref{fig:path-enumeration} for an illustration.
As soon as the shortest path exceeds $\rho \cdot d(s,t)$ or does not exist, we have found a suitable $E'$.

\begin{figure}
    \centering
    \includegraphics{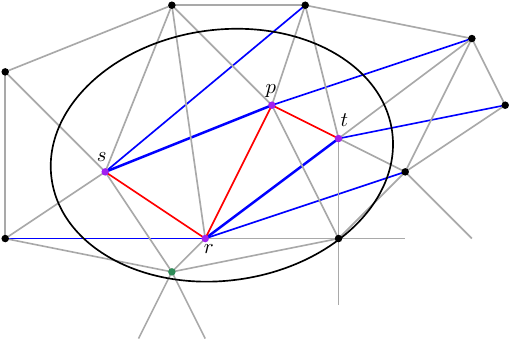}
    \caption{Sketch of the dilation path separation between $s$ and $t$;
             the red path represents the shortest $s$-$t$-path in the current triangulation $T$.
             The ellipse represents the points with total distance $\rho \cdot d(s,t)$ to $s,t$.
             Red and gray edges represent the current triangulation $T$ and are a subset of $E$,
             blue edges are possible edges in $E$, but not in $T$.
             The green point is found from $s$ but eliminated when running Dijkstra from $t$;
             the purple points are the points in $Q$ that may be on sufficiently short $s$-$t$-paths.
             The clause we generate contains the bold blue edges $sp$ and $rt$; 
             without these edges, there is no sufficiently short $s$-$t$-path in $E$.}
    \label{fig:path-enumeration}
\end{figure}

In practice, we have no efficient oracle for square root sum comparisons;
see our remarks on computing exact shortest paths in \cref{sec:exactness-implementation-issues} for details on how we handle this.

\subsubsection{Incremental SAT solving}
\label{sec:incremental-sat-solving}
Many modern SAT solvers are incremental, at least to some extend.
Among other things, this means that we can add clauses to the solver after solving a formula and check whether the new set of clauses is satisfiable.
This is often much more efficient than reconstructing the solver and re-solving the formula from scratch.
Among other things, the solver gets to keep learned clauses and often also parts of the current state of the search;
we see this reflected in the fact that, when we add clauses in \incmdt{}, the majority of the triangulation usually remains unchanged
even though the solver could theoretically change large parts of the triangulation.

However, we must be more careful in \binmdt{}.
As soon as we have used an incremental SAT solver on a bound that cannot be achieved, we have to construct a new solver object.
This is because clauses may have been added to the solver that are not valid for less tight bounds.
We therefore maintain a pool of clauses that we have generated;
this allows us to quickly reconstruct a new SAT solver without losing such clauses.
All clauses that we generate during the current iteration are only added to that pool
once we have found a triangulation that satisfies the current bound.

We could alternatively generate tentative clauses by adding an \emph{assumption literal} $a_i$ on a new variable to all clauses generated by the $i$th iteration of \binmdt{}.
We could then tell the incremental solver to temporarily \emph{assume} that $a_i$ is false, which would effectively enforce the clauses.
If we find a triangulation satisfying the bound, we can make this assumption permanent by adding $\lnot a_i$ as clause.
If we instead find the bound to be too tight, we can disable the clauses by removing the assumption and adding $a_i$ as clause.

We did not study this aspect further, since the time spent in the SAT solver was negligible compared to the time spent computing dilations.

\subsection{Exactness}
\label{sec:exactness-implementation-issues}
Ultimately, we aim for solutions that are exactly optimal, including all numerical issues, as if we were using infinite precision for all our computations.
In this section, we describe the practical challenges this presents and our approach to handling them.

\subsubsection{Dilation and shortest paths}
Achieving this requires us to exactly compute the dilation.
We thus rely on an exact implementation of Dijkstra's algorithm for Euclidean distances.
We run Dijkstra's algorithm once for each source point $p$, computing the ratio between the shortest path and the Euclidean distance for $p$ and all other points.
The theoretical runtime for this is $O(n^2\log n)$ assuming uniform cost for each computational operation; we parallelize the runs for different source points.
We represent each dilation value as a combination of a path defining the dilation and an interval containing the true dilation value.
In most cases, interval arithmetic is sufficient to decide which path defines the dilation;
however, in particular in some non-randomly generated instances that include symmetries, we encounter cases where intervals do not suffice.
In fact, the path defining the dilation is often non-unique in these instances.

In a first prototype, we simply resorted to directly constructing an instance of CGAL's exact number type with support for square roots, representing the precise dilation values in order to compare them if interval arithmetic was insufficient.
However, the comparisons resulted in unacceptable performance issues; it turned out to be crucial to perform extensive preprocessing before resorting to such expensive operations:
even single comparisons of sums consisting of less than $10$ square roots each often took several seconds or minutes.

Let $\pi_1 = (p_1, \ldots, p_k)$ and $\pi_2 = (q_1, \ldots, q_m)$ be two paths for which we want to decide which one incurs a worse detour.
In other words, we want to determine whether
\[\frac{\sum\limits_{i = 2}^k d(p_{i-1}, p_i)}{d(p_1, p_k)} = \underbrace{\sum\limits_{i = 2}^k \frac{d(p_{i-1}, p_i)}{d(p_1, p_k)}}_{(*)} \overset{?}{<} \frac{\sum\limits_{i = 2}^m d(q_{i-1}, q_i)}{d(q_1, q_m)} = \underbrace{\sum\limits_{i = 2}^m \frac{d(q_{i-1}, q_i)}{d(q_1, q_m)}}_{(**)}.\]
Observe that each entry of the sums $(*)$ and $(**)$ is the square root of a rational number;
we can compute and compare these rational numbers at much lower cost than dealing with the square roots.
This allows us to eliminate equal terms from the sums $(*)$ and $(**)$.
We also scan for horizontal or vertical segments, which have rational length for rational input coordinates,
and collect them in a single term on each side.

We use these rational numbers to compute more precise intervals for each of the terms,
summing up $(*) - (**)$ in interval arithmetic in order of increasing absolute value.
If the resulting interval does not contain $0$, we can decide which path has higher dilation.
If it is precisely $[0,0]$, e.g., because we eliminated all terms, we know both paths have the same dilation.
Otherwise, we repeat the computation using intervals of 1024-bit mantissa floating point values using the MPFR~\cite{DBLP:journals/toms/FousseHLPZ07} library instead of double-precision values.
Only if all these steps fail do we construct the dilation values using CGAL's exact number type and compare them;
in all cases we observed, at this point, the resulting values were actually equal.
We apply similar techniques to compare the length of paths instead of dilations in our implementation of Dijkstra's algorithm.

\subsubsection{Candidate edge enumeration}
Not only the computation of shortest paths and the dilation require attention to numerical issues,
but these are the only cases necessarily involving the comparison of sums of square roots.
The incremental search procedure is also exact, i.e., events regarding points and nodes are handled in precisely the correct order,
while activation distance events are handled based on upper bounds on the exact activation distance.
We order all events by squared distances instead of distances and thus maintain rationality,
allowing us to augment interval arithmetic with rational computations to ultimately decide the order of events.
Most other parts of the enumeration process only require exact predicates, such as intersections checks;
these are provided by CGAL and also make use of interval arithmetic to only resort to exact computations if necessary.

Another situation where the use of interval arithmetic is important and not straightforward is the representation of dead sectors.
Since working with true polar angles involves the costly trigonometric function \texttt{atan2},
which is not readily available with correct rounding for use in interval arithmetic,
we instead use \emph{pseudo-angles} as described by Moret and Shapiro~\cite{10.5555/102912}.
In essence, we measure angles by computing distances along the $\ell_1$ unit circle instead of the $\ell_2$ unit circle.
Pseudo-angles are rational and only use operations that are available in correctly-rounded fashion on modern CPUs;
this allows us to use interval arithmetic to handle the polar angle intervals of dead sectors.
In order to allow us to merge polar angle intervals that touch in a single point,
which would not be possible using interval arithmetic alone, we also store, for each polar angle interval,
the lowest and highest point (w.r.t.\ the pseudo-angle around $p$) in the interval.

\section{Empirical evaluation}

In this section, we provide additional details on the experiments conducted in the context of the empirical evaluation,
together with additional figures and data tables; 
note that all experiment data are archived and publicly accessible; see the supplement.

\begin{description}
    \item[QA1] How does the runtime of the dilation computation scale with the number of points in the input instance?
\end{description}

\subsection{QA1: Dilation computation}
\label{sec:experiments-dilation-computation}

\begin{figure}
    \begin{subfigure}[t]{0.49\linewidth}
        \centering
        \includegraphics[width=\linewidth]{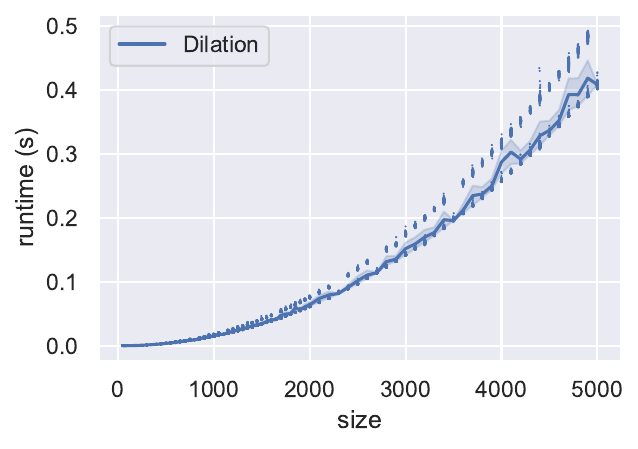}
    \end{subfigure}\hfill
    \begin{subfigure}[t]{0.49\linewidth}
        \centering
        \includegraphics[width=\linewidth]{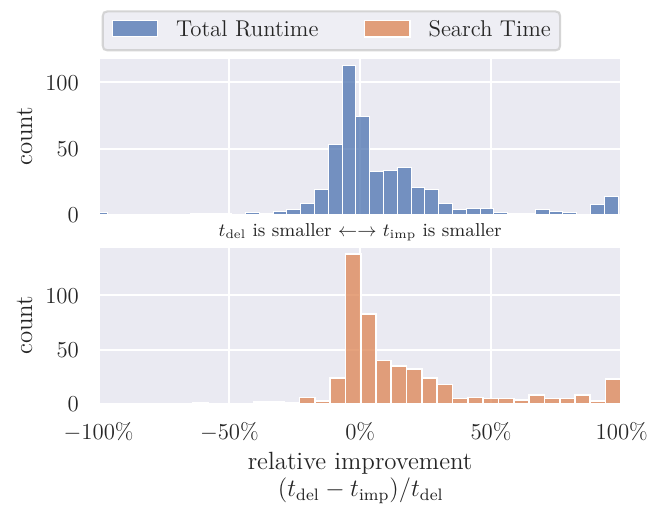}
    \end{subfigure}
    \caption{\textbf{(Left)}~The dilation computation time increases significantly for larger instances. 
             \textbf{(Right)}~Runtime impact of working with the Delaunay triangulation ($t_{\operatorname{del}}$) directly vs.~using the initial improvement algorithm ($t_{\operatorname{imp}}$).
                     Introducing initial improvements reduces the search time of \binmdt{} significantly but can come at the cost of a higher overall runtime.}
    \label{fig:dilation-computation}
    \label{fig:initial-improvement}
\end{figure}

The proposed algorithms rely on the computation of the dilation for some given triangulation.
Our implementation uses an exact version of a bidirectional Dijkstra algorithm to compute the dilation in $O(n^2 \log n)$, see~\cref{sec:exactness-implementation-issues} for details.
The experiment was conducted on the \emph{random} instance set with up to \num{5000} points where the input triangulation is the Delaunay triangulation.
To obtain a reasonable average runtime, we computed the dilation for each instance \num{100} times, see~\cref{fig:dilation-computation}.
A dilation query for an instance with \num{5000} points takes on average slightly more than \qty{0.4}{s}.
As expected, the runtime increases with the number of points in the instance with a quadratic dependency.
In our experiments on the \emph{public} instance set we observed that the runtime of the dilation computation can be significantly higher. This is due to degenerate instances containing structures like regular grids that have a lot of edges and paths with the same length.
Doing exact comparisons of square roots in these cases can be costly and lead to higher runtimes.



\begin{table}[h!t]
    \centering
    \input{tables/tsplib}
    \caption{Solutions and runtimes for small TSPLIB instances.}
    \label{tab:tsplib-comparison}
\end{table}

\subsection{Q4: Algorithm comparison}
\label{sec:appendix-algorithm-comparison}

This section contains additional details on the comparison of the proposed algorithms to each other and should be seen as an extension to~\cref{sec:experiments-algorithm-comparison}.
\cref{fig:experiments-public-instance-set} shows the runtime of the \binmdt{} with two different initial solution strategies on the \emph{public} benchmark set with up to \num{10000} points.

We provide an additional information in~\cref{fig:initial-improvement} that shows the relative runtime change when using the different initial solution strategies.
Let $t_{\mathrm{del}}$ be the runtime of \binmdt{} when starting with the Delaunay triangulation as initial solution and $t_{\mathrm{imp}}$ be the runtime when starting with the improvement Delaunay triangulation.
We distinguish between the overall runtime (including calculating the initial solution) and the search time (the time spent in finding the optimal solution).
What we compare is the relative improvement, i.e. $\frac{t_{\mathrm{del}}-t_{\mathrm{imp}}}{t_{\mathrm{del}}}$.
A positive value indicates that the improved Delaunay triangulation reduces the runtime and a negative value indicates that the initial improvement increases the runtime.
Regarding the search time we observe that the improved solution reduces the time it takes to find and verify the optimal solution. This is expected as the heuristic is known to produce better solutions than the Delaunay triangulation and in fact can never be worse than the Delaunay triangulation, see~\cref{sec:experiments-initial-solutions}.
The overall runtime however can sometimes be higher when using the improved triangulation as the initial solution. This is due to the fact that testing different constrained Delaunay triangulations require multiple dilation computations which can be costly.
Still we conclude that using the improved initial solution is beneficial in most cases and is used as the default setting for the \binmdt{} algorithm.

\section{Additional remarks on regular \texorpdfstring{$n$}{n}-gons}
\label{sec:appendix-n-gons}
It was already noted by Euler (see~\cite{stanley2015catalan}) that the number of triangulations of convex $n$-gons corresponds to the Catalan numbers
$C_{n-2}$, which grow asymptotically like $\frac{4^n}{n^{3/2}\sqrt{\pi}}$, illustrating the futility of performing 
a brute-force enumeration for finding an MDT. 

In \cref{sec:n-gon-lb} we answer the open question posed by Dumitrescu and Ghosh~\cite[Problem 1]{DBLP:journals/ijcga/DumitrescuG16} that originated from Bose, Smit and Kanj~\cite{DBLP:journals/comgeo/BoseS13, DBLP:conf/iccit/Kanj13}.
With the help of our exact algorithm we can obtain a lower bound for the dilation of regular $n$-gons.
It is known that the worst-case dilation of a regular $n$-gon is between $1.4308$ and $1.4482$.

\begin{theorem}[\cite{DBLP:journals/ijcga/DumitrescuG16,DBLP:journals/comgeo/SattariI19}]
    \label{theorem:old-bounds}
    The worst-case dilation of a regular $n$-gon $\rho^*$ is bounded by
    \[
        1.4308 \approx 
            \frac{2\sin(\frac{2\pi}{23}) + \sin(\frac{8\pi}{23})}{\sin(\frac{11\pi}{23})} \leq \rho^* < 1.4482.
    \]
\end{theorem}

\begin{figure}[t!p]
    \centering
    \includegraphics[width=\linewidth]{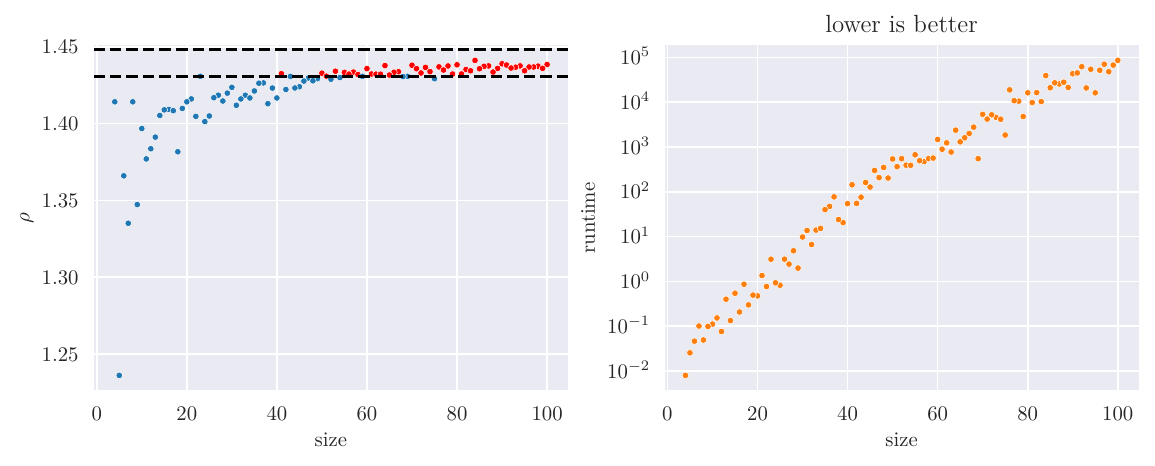}
    \caption{Dilations and runtimes for regular $n$-gons for $4 \leq n \leq 100$.
             Red dots improve the current lower bound on the dilation of regular $n$-gons of $1.4308$ that comes from the regular $23$-gon.
             The dashed black lines mark the known upper bound of $1.4482$ and the previous best lower bound of $1.4308$.}
    \label{fig:n-gon-rho-runtimes}
\end{figure}

\begin{figure}[t!p]
    \centering
    \includegraphics[width=\linewidth]{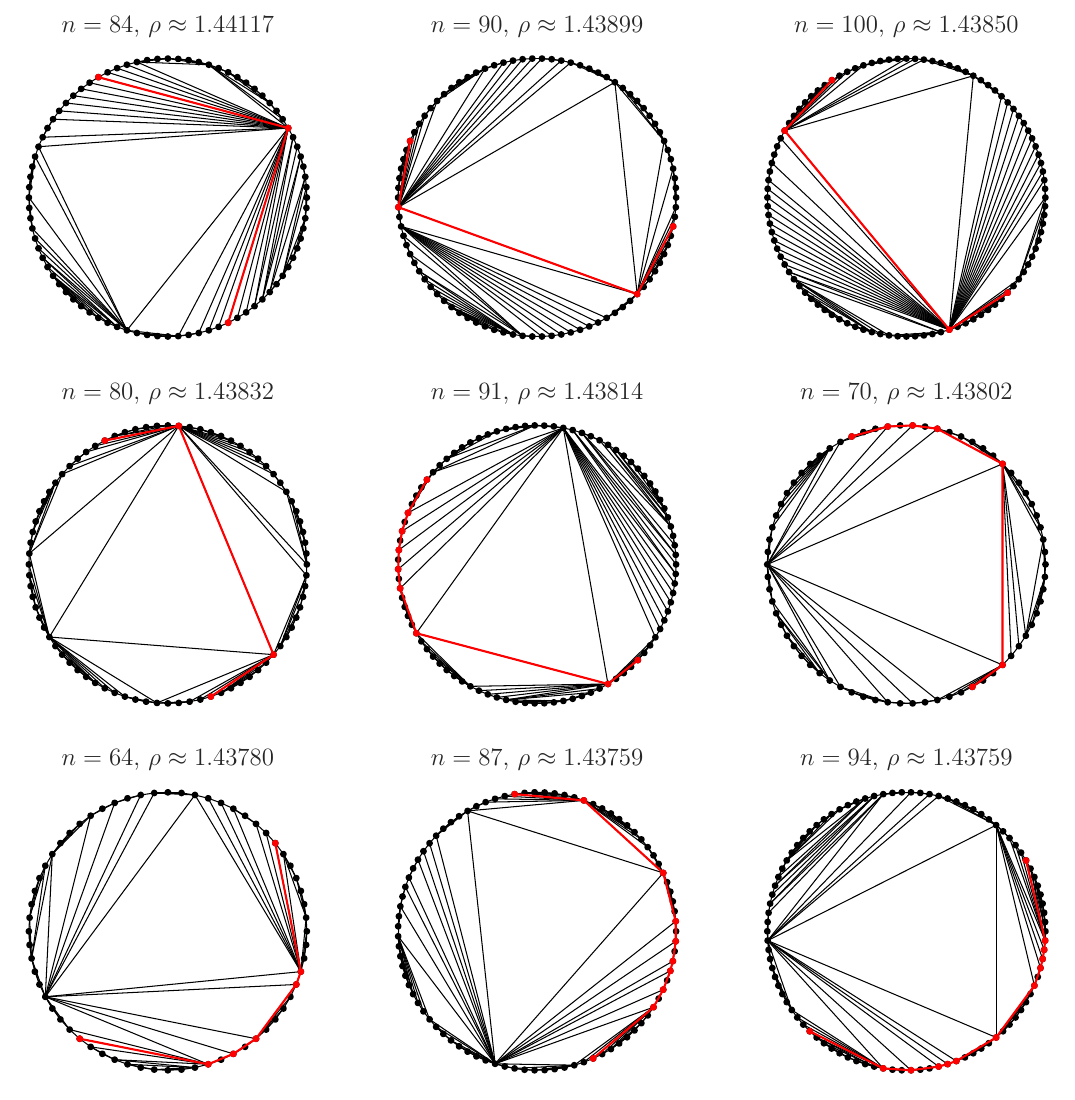}
    \caption{Optimal MDT of regular $n$-gon with floating point coordinates. All solutions improve the current lower bound for regular $n$-gons.} 
    \label{fig:n-gon-examples}
\end{figure}

We ran our exact algorithm for regular $n$-gons for $4 \leq n \leq 100$. 
\Cref{fig:n-gon-rho-runtimes} highlights the previous bounds
from~\cref{theorem:old-bounds} and shows that many regular $n$-gons improve the
current lower bound of $1.4308$. 
\Cref{fig:n-gon-examples} shows some examples of such solutions. 
One can observe that the solution structures appear to be quite intricate, 
highlighting that it would have been quite challenging
to establish these results with only the help of manual analysis. Moreover,
this intricacy may be an indication that the underlying problem does not
lend itself to a simple polynomial-time algorithm.

These difficulties are also visible in the necessary runtimes for $n$-gons.
It is remarkable that the growth (which appears to be exponential in $n$)
is significantly faster than for our results on 
instances from the \emph{uniform} and \emph{public} benchmark sets (see~\cref{sec:experiments-algorithm-comparison}),
with already more than \qty{10}{h} for $n = 84$ points.
(See~\cref{fig:n-gon-rho-runtimes} and~\cref{tab:n-gons} for detailed results.)
An intuitive reason may be that any point in the interior can be used as a Steiner point for shortcuts, 
while also allowing local reduction techniques such as the ellipse property, which fails
to provide a reduction for the vertices of a regular $n$-gon. Conversely, any chord
through the interior constitutes an obstacle for all separated point pairs, producing
cascading sets of constraints. Moreover, 
any solution contains many symmetries that make finding an optimal solution (among many others with similar or same dilation) significantly harder.
As all points are on the boundary of a circle, quadtree queries in the enumeration algorithm are not as effective as in other instance sets.
Finally, regular $n$-gons contain a maximum number of intersections between edges, which makes intersection constraints weaker.
It is this combination of difficulties that make regular $n$-gons all the more interesting,
as they are prime candidates for the worst-case dilation bound.

\begin{conjecture}
\label{con:worst}
The worst-case bound for Minimum Dilation Triangulations corresponds to the 
asymptotic value for regular $n$-gons. This implies that the correct 
value is $1.44\ldots$
\end{conjecture}

\begin{table}[t!p]
    \centering
    \input{tables/n_gons}
    \caption{Dilations of regular $n$-gons for $4 \leq n \leq 100$ and the runtime of the \binmdt{} algorithm. Bold values improve the current lower bound on the dilation of regular $n$-gons of $1.4308$ that arises from the regular $23$-gon.}
    \label{tab:n-gons}
\end{table}

\clearpage

\section{Additional tables}
Here we provide a full overview of additional experiment data in table form.

\begin{table}[h!]
    \centering
    \input{tables/large_benchmark}
    \caption{Solutions and runtimes of \binmdt{} for large instances of the \emph{public} benchmark set. Columns also provide the number of full and sampled dilation computations that are the bottleneck in the exact algorithm.}
    \label{tab:large-comparison}
\end{table}

\clearpage

\input{tables/public_instances}

%% file: tables/tsplib.tex
\begin{tabular}{|lccc|}
    \hline
    instance & $\rho$ & BnB~\cite{DBLP:journals/jgo/SattariI17} runtime (s)  & \binmdt{} runtime (s) \\
    \hline
    burma14 & \num{1.1749} & \num{0.0040} & \num{0.0061} \\
    ulysses16 & \num{1.1782} & \num{0.0060} & \num{0.0010} \\
    ulysses22 & \num{1.1782} & \num{0.0230} & \num{0.0077} \\
    att48 & \num{1.3279} & \num{0.6570} & \num{0.0098} \\
    eil51 & \num{1.3306} & \num{0.6990} & \num{0.0070} \\
    berlin52 & \num{1.2693} & \num{1.4140} & \num{0.0097} \\
    st70 & \num{1.3002} & \num{8.5910} & \num{0.0099} \\
    eil76 & \num{1.3407} & \num{4.3370} & \num{0.0018} \\
    pr76 & \num{1.3013} & \num{26.7950} & \num{0.0340} \\
    gr96 & \num{1.3310} & \num{20.4990} & \num{0.0031} \\
    rat99 & \num{1.3280} & \num{31.1510} & \num{0.0133} \\
    kroA100 & \num{1.3386} & \num{18.3480} & \num{0.0027} \\
    kroB100 & \num{1.2910} & \num{18.4390} & \num{0.0097} \\
    kroC100 & \num{1.3263} & \num{17.6490} & \num{0.0215} \\
    kroD100 & \num{1.3159} & \num{19.3270} & \num{0.0153} \\
    kroE100 & \num{1.3416} & \num{21.0360} & \num{0.0270} \\
    rd100 & \num{1.3328} & \num{18.7420} & \num{0.0180} \\
    eil101 & \num{1.4142} & \num{14.8230} & \num{0.0153} \\
    lin105 & \num{1.3116} & \num{36.7820} & \num{0.0217} \\
    pr107 & \num{1.2504} & \num{55.8490} & \num{0.1261} \\
    pr124 & \num{1.3077} & \num{67.2010} & \num{0.0305} \\
    bier127 & \num{1.2999} & \num{47.7660} & \num{0.0231} \\
    ch130 & \num{1.3535} & \num{61.8330} & \num{0.0549} \\
    pr136 & \num{1.4037} & \num{159.0230} & \num{0.0066} \\
    gr137 & \num{1.2732} & \num{126.4090} & \num{0.0146} \\
    pr144 & \num{1.2469} & \num{1248.7270} & \num{0.2512} \\
    ch150 & \num{1.3032} & \num{112.2890} & \num{0.0195} \\
    kroA150 & \num{1.3386} & \num{107.5780} & \num{0.0038} \\
    kroB150 & \num{1.3217} & \num{113.8820} & \num{0.0169} \\
    pr152 & \num{1.2770} & \num{547.4960} & \num{0.2999} \\
    u159 & \num{1.4142} & \num{158.4300} & \num{0.0501} \\
    rat195 & \num{1.3436} & \num{610.1950} & \num{0.0274} \\
    d198 & \num{1.4142} & \num{1131.8730} & \num{0.0124} \\
    kroA200 & \num{1.3863} & \num{397.6640} & \num{0.0572} \\
    kroB200 & \num{1.3804} & \num{398.1610} & \num{0.0055} \\
    \hline
    \end{tabular}

%% file: tables/n_gons.tex
\begin{tabular}{|lcc||lcc||lcc|}
    \hline
    n & $\rho$ & runtime (s) & n & $\rho$ & runtime (s) & n & $\rho$ & runtime (s) \\
    \hline
    \num{4} & $1.41421$ & $<1$ & \num{37} & $1.42651$ & \num{77} & \num{70} & $\mathbf{1.43802}$ & \num{5357} \\
    \num{5} & $1.23606$ & $<1$ & \num{38} & $1.41300$ & \num{24} & \num{71} & $\mathbf{1.43579}$ & \num{4190} \\
    \num{6} & $1.36602$ & $<1$ & \num{39} & $1.42316$ & \num{21} & \num{72} & $\mathbf{1.43298}$ & \num{5236} \\
    \num{7} & $1.33512$ & $<1$ & \num{40} & $1.41671$ & \num{55} & \num{73} & $\mathbf{1.43659}$ & \num{4574} \\
    \num{8} & $1.41421$ & $<1$ & \num{41} & $\mathbf{1.43250}$ & \num{144} & \num{74} & $\mathbf{1.43388}$ & \num{4160} \\
    \num{9} & $1.34729$ & $<1$ & \num{42} & $1.42222$ & \num{55} & \num{75} & $1.42925$ & \num{1850} \\
    \num{10} & $1.39680$ & $<1$ & \num{43} & $1.43074$ & \num{76} & \num{76} & $\mathbf{1.43695}$ & \num{18849} \\
    \num{11} & $1.37703$ & $<1$ & \num{44} & $1.42320$ & \num{161} & \num{77} & $\mathbf{1.43490}$ & \num{10805} \\
    \num{12} & $1.38366$ & $<1$ & \num{45} & $1.42403$ & \num{128} & \num{78} & $\mathbf{1.43755}$ & \num{10536} \\
    \num{13} & $1.39121$ & $<1$ & \num{46} & $1.42771$ & \num{300} & \num{79} & $\mathbf{1.43228}$ & \num{4788} \\
    \num{14} & $1.40532$ & $<1$ & \num{47} & $1.42969$ & \num{208} & \num{80} & $\mathbf{1.43832}$ & \num{16235} \\
    \num{15} & $1.40897$ & $<1$ & \num{48} & $1.42791$ & \num{350} & \num{81} & $\mathbf{1.43232}$ & \num{9818} \\
    \num{16} & $1.40924$ & $<1$ & \num{49} & $1.42927$ & \num{203} & \num{82} & $\mathbf{1.43525}$ & \num{16369} \\
    \num{17} & $1.40844$ & $<1$ & \num{50} & $\mathbf{1.43283}$ & \num{541} & \num{83} & $\mathbf{1.43441}$ & \num{10314} \\
    \num{18} & $1.38169$ & $<1$ & \num{51} & $\mathbf{1.43089}$ & \num{365} & \num{84} & $\mathbf{1.44116}$ & \num{39277} \\
    \num{19} & $1.40988$ & $<1$ & \num{52} & $1.42892$ & \num{549} & \num{85} & $\mathbf{1.43574}$ & \num{21030} \\
    \num{20} & $1.41421$ & $<1$ & \num{53} & $\mathbf{1.43411}$ & \num{393} & \num{86} & $\mathbf{1.43731}$ & \num{27062} \\
    \num{21} & $1.41610$ & \num{1} & \num{54} & $1.42999$ & \num{391} & \num{87} & $\mathbf{1.43759}$ & \num{25565} \\
    \num{22} & $1.40471$ & $<1$ & \num{55} & $\mathbf{1.43335}$ & \num{672} & \num{88} & $\mathbf{1.43357}$ & \num{27903} \\
    \num{23} & $1.43081$ & \num{3} & \num{56} & $\mathbf{1.43224}$ & \num{495} & \num{89} & $\mathbf{1.43592}$ & \num{21326} \\
    \num{24} & $1.40133$ & $<1$ & \num{57} & $\mathbf{1.43356}$ & \num{474} & \num{90} & $\mathbf{1.43898}$ & \num{43300} \\
    \num{25} & $1.40497$ & $<1$ & \num{58} & $\mathbf{1.43199}$ & \num{551} & \num{91} & $\mathbf{1.43814}$ & \num{45019} \\
    \num{26} & $1.41690$ & \num{3} & \num{59} & $1.43077$ & \num{567} & \num{92} & $\mathbf{1.43620}$ & \num{62126} \\
    \num{27} & $1.41850$ & \num{2} & \num{60} & $\mathbf{1.43582}$ & \num{1472} & \num{93} & $\mathbf{1.43675}$ & \num{20787} \\
    \num{28} & $1.41472$ & \num{5} & \num{61} & $\mathbf{1.43230}$ & \num{890} & \num{94} & $\mathbf{1.43758}$ & \num{54273} \\
    \num{29} & $1.41986$ & \num{2} & \num{62} & $\mathbf{1.43221}$ & \num{1237} & \num{95} & $\mathbf{1.43442}$ & \num{16205} \\
    \num{30} & $1.42362$ & \num{10} & \num{63} & $\mathbf{1.43217}$ & \num{772} & \num{96} & $\mathbf{1.43687}$ & \num{51445} \\
    \num{31} & $1.41196$ & \num{14} & \num{64} & $\mathbf{1.43779}$ & \num{2376} & \num{97} & $\mathbf{1.43687}$ & \num{70321} \\
    \num{32} & $1.41607$ & \num{7} & \num{65} & $\mathbf{1.43164}$ & \num{1302} & \num{98} & $\mathbf{1.43743}$ & \num{47885} \\
    \num{33} & $1.41845$ & \num{14} & \num{66} & $\mathbf{1.43343}$ & \num{1609} & \num{99} & $\mathbf{1.43599}$ & \num{67183} \\
    \num{34} & $1.41671$ & \num{15} & \num{67} & $\mathbf{1.43381}$ & \num{1995} & \num{100} & $\mathbf{1.43849}$ & \num{85588} \\
    \num{35} & $1.42127$ & \num{40} & \num{68} & $1.43063$ & \num{2773} &  &  &  \\
    \num{36} & $1.42630$ & \num{47} & \num{69} & $1.43078$ & \num{550} &  &  &  \\
    \hline
\end{tabular}

%% file: tables/large_benchmark.tex
\begin{tabular}{|lrrrrr|}
    \hline
    instance & $n$ & $\rho$ & \#full dil. & \#sampled dil. & \binmdt{} time (s) \\
    \hline
    xmc10150 & 10150 & \num{1.4142} & \num{1} & \num{22} & \num{6} \\
rl11849 & 11849 & \num{1.4142} & \num{1} & \num{10} & \num{4} \\
usa13509 & 13509 & \num{1.3882} & \num{307} & \num{451} & \num{735} \\
xvb13584 & 13584 & \num{1.4142} & \num{1} & \num{87} & \num{37} \\
brd14051 & 14051 & \num{1.4142} & \num{1} & \num{21} & \num{6} \\
xrb14233 & 14233 & \num{1.4142} & \num{1} & \num{15} & \num{5} \\
d15112 & 15112 & \num{1.3852} & \num{105} & \num{199} & \num{332} \\
xia16928 & 16928 & \num{1.4142} & \num{1} & \num{21} & \num{9} \\
pjh17845 & 17845 & \num{1.4142} & \num{1} & \num{18} & \num{8} \\
d18512 & 18512 & \num{1.4142} & \num{86} & \num{200} & \num{502} \\
frh19289 & 19289 & \num{1.4142} & \num{1} & \num{91} & \num{59} \\
fnc19402 & 19402 & \num{1.4142} & \num{1} & \num{8} & \num{4} \\
pntset-0020000 & 20000 & \num{1.4018} & \num{229} & \num{434} & \num{1327} \\
fpg-poly-0000020000 & 20000 & \num{1.3825} & \num{975} & \num{1208} & \num{4664} \\
us-night-0020000 & 20000 & \num{1.3948} & \num{821} & \num{1037} & \num{4098} \\
protein-0020000 & 20000 & \num{1.3949} & \num{780} & \num{1052} & \num{3952} \\
uniform-0020000-2 & 20000 & \num{1.3861} & \num{291} & \num{508} & \num{1548} \\
euro-night-0020000 & 20000 & \num{1.3936} & \num{971} & \num{1250} & \num{4990} \\
uniform-0020000-1 & 20000 & \num{1.3903} & \num{290} & \num{537} & \num{1654} \\
jupiter-0020000 & 20000 & \num{1.3829} & \num{3422} & \num{3934} & \num{16384} \\
ido21215 & 21215 & \num{1.4142} & \num{1} & \num{14} & \num{7} \\
fma21553 & 21553 & \num{1.4142} & \num{1} & \num{62} & \num{61} \\
lsb22777 & 22777 & \num{1.4142} & \num{1} & \num{41} & \num{29} \\
xrh24104 & 24104 & \num{1.4142} & \num{1} & \num{78} & \num{78} \\
bbz25234 & 25234 & \num{1.4142} & \num{1} & \num{26} & \num{17} \\
irx28268 & 28268 & \num{1.4142} & \num{1} & \num{39} & \num{26} \\
fyg28534 & 28534 & \num{1.4142} & \num{1} & \num{38} & \num{34} \\
icx28698 & 28698 & \num{1.4142} & \num{1675} & \num{2026} & \num{32372} \\
boa28924 & 28924 & \num{1.4142} & \num{1} & \num{32} & \num{25} \\
ird29514 & 29514 & \num{1.4142} & \num{1} & \num{45} & \num{52} \\
protein-0030000 & 30000 & \num{1.3931} & \num{1798} & \num{2168} & \num{19527} \\
us-night-0030000 & 30000 & \num{1.3968} & \num{1683} & \num{1972} & \num{18160} \\
euro-night-0030000 & 30000 & \num{1.3976} & \num{1204} & \num{1506} & \num{13871} \\
uniform-0030000-1 & 30000 & \num{1.3946} & \num{322} & \num{615} & \num{4076} \\
pntset-0030000 & 30000 & \num{1.3929} & \num{423} & \num{820} & \num{5635} \\
fpg-poly-0000030000 & 30000 & \num{1.3885} & \num{1042} & \num{1341} & \num{11842} \\
uniform-0030000-2 & 30000 & \num{1.3870} & \num{414} & \num{779} & \num{5391} \\
jupiter-0030000 & 30000 & \num{1.3785} & \num{3394} & \num{6129} & \num{58752} \\
    \hline
    \end{tabular}

%% file: tables/public_instances.tex
\begin{longtable}{|lrrlll|}
    \hline
    instance & $n$ & $\rho$ & \#full dil. & \#sampled dil. & \binmdt{} time (s) \\
    \hline
    \endhead
    london-0000010 & 10 & \num{1.1010} & \num{1} & \num{7} & $<1$ \\
stars-0000010 & 10 & \num{1.1953} & \num{1} & \num{5} & $<1$ \\
pntset-0000010 & 10 & \num{1.2032} & \num{1} & \num{7} & $<1$ \\
euro-night-0000010 & 10 & \num{1.2184} & \num{2} & \num{2} & $<1$ \\
uniform-0000010-1 & 10 & \num{1.2528} & \num{2} & \num{4} & $<1$ \\
fpg-poly-0000000010 & 10 & \num{1.0811} & \num{1} & \num{6} & $<1$ \\
uniform-0000010-2 & 10 & \num{1.2215} & \num{1} & \num{9} & $<1$ \\
us-night-0000010 & 10 & \num{1.1920} & \num{1} & \num{1} & $<1$ \\
burma14 & 14 & \num{1.1749} & \num{1} & \num{8} & $<1$ \\
stars-0000015 & 15 & \num{1.1368} & \num{1} & \num{5} & $<1$ \\
london-0000015 & 15 & \num{1.2564} & \num{1} & \num{8} & $<1$ \\
uniform-0000015-2 & 15 & \num{1.2225} & \num{1} & \num{7} & $<1$ \\
uniform-0000015-1 & 15 & \num{1.2434} & \num{1} & \num{7} & $<1$ \\
euro-night-0000015 & 15 & \num{1.1975} & \num{1} & \num{1} & $<1$ \\
us-night-0000015 & 15 & \num{1.2868} & \num{2} & \num{8} & $<1$ \\
ulysses16 & 16 & \num{1.1782} & \num{1} & \num{1} & $<1$ \\
uniform-0000020-1 & 20 & \num{1.2368} & \num{1} & \num{8} & $<1$ \\
uniform-0000020-2 & 20 & \num{1.3345} & \num{1} & \num{8} & $<1$ \\
fpg-poly-0000000020 & 20 & \num{1.3152} & \num{1} & \num{8} & $<1$ \\
pntset-0000020 & 20 & \num{1.2739} & \num{1} & \num{8} & $<1$ \\
us-night-0000020 & 20 & \num{1.3101} & \num{1} & \num{7} & $<1$ \\
euro-night-0000020 & 20 & \num{1.2824} & \num{1} & \num{8} & $<1$ \\
london-0000020 & 20 & \num{1.2747} & \num{1} & \num{8} & $<1$ \\
stars-0000020 & 20 & \num{1.1946} & \num{1} & \num{6} & $<1$ \\
ulysses22 & 22 & \num{1.1782} & \num{1} & \num{7} & $<1$ \\
euro-night-0000025 & 25 & \num{1.2483} & \num{1} & \num{7} & $<1$ \\
uniform-0000025-1 & 25 & \num{1.2888} & \num{1} & \num{1} & $<1$ \\
stars-0000025 & 25 & \num{1.2193} & \num{1} & \num{6} & $<1$ \\
us-night-0000025 & 25 & \num{1.3488} & \num{1} & \num{9} & $<1$ \\
uniform-0000025-2 & 25 & \num{1.3500} & \num{1} & \num{8} & $<1$ \\
london-0000025 & 25 & \num{1.2960} & \num{2} & \num{6} & $<1$ \\
london-0000030 & 30 & \num{1.3355} & \num{1} & \num{1} & $<1$ \\
pntset-0000030 & 30 & \num{1.2720} & \num{2} & \num{7} & $<1$ \\
fpg-poly-0000000030 & 30 & \num{1.2313} & \num{2} & \num{7} & $<1$ \\
euro-night-0000030 & 30 & \num{1.2627} & \num{1} & \num{1} & $<1$ \\
us-night-0000030 & 30 & \num{1.2997} & \num{2} & \num{3} & $<1$ \\
uniform-0000030-2 & 30 & \num{1.3112} & \num{2} & \num{10} & $<1$ \\
uniform-0000030-1 & 30 & \num{1.2722} & \num{1} & \num{7} & $<1$ \\
stars-0000030 & 30 & \num{1.2653} & \num{1} & \num{1} & $<1$ \\
us-night-0000035 & 35 & \num{1.2630} & \num{1} & \num{11} & $<1$ \\
stars-0000035 & 35 & \num{1.3488} & \num{2} & \num{9} & $<1$ \\
uniform-0000035-2 & 35 & \num{1.2930} & \num{1} & \num{8} & $<1$ \\
uniform-0000035-1 & 35 & \num{1.2591} & \num{3} & \num{9} & $<1$ \\
london-0000035 & 35 & \num{1.2363} & \num{1} & \num{6} & $<1$ \\
euro-night-0000035 & 35 & \num{1.3131} & \num{1} & \num{1} & $<1$ \\
stars-0000040 & 40 & \num{1.2998} & \num{1} & \num{5} & $<1$ \\
london-0000040 & 40 & \num{1.2877} & \num{2} & \num{5} & $<1$ \\
pntset-0000040 & 40 & \num{1.2885} & \num{1} & \num{7} & $<1$ \\
us-night-0000040 & 40 & \num{1.2897} & \num{1} & \num{6} & $<1$ \\
euro-night-0000040 & 40 & \num{1.3022} & \num{1} & \num{1} & $<1$ \\
fpg-poly-0000000040 & 40 & \num{1.2953} & \num{1} & \num{7} & $<1$ \\
uniform-0000040-2 & 40 & \num{1.2608} & \num{3} & \num{8} & $<1$ \\
uniform-0000040-1 & 40 & \num{1.3211} & \num{1} & \num{7} & $<1$ \\
uniform-0000045-1 & 45 & \num{1.3216} & \num{2} & \num{7} & $<1$ \\
us-night-0000045 & 45 & \num{1.2954} & \num{1} & \num{1} & $<1$ \\
uniform-0000045-2 & 45 & \num{1.2830} & \num{2} & \num{5} & $<1$ \\
london-0000045 & 45 & \num{1.2705} & \num{1} & \num{5} & $<1$ \\
stars-0000045 & 45 & \num{1.2834} & \num{2} & \num{7} & $<1$ \\
euro-night-0000045 & 45 & \num{1.2839} & \num{2} & \num{4} & $<1$ \\
att48 & 48 & \num{1.3279} & \num{1} & \num{8} & $<1$ \\
london-0000050 & 50 & \num{1.2953} & \num{1} & \num{7} & $<1$ \\
uniform-0000050-2 & 50 & \num{1.3042} & \num{1} & \num{5} & $<1$ \\
stars-0000050 & 50 & \num{1.3652} & \num{1} & \num{8} & $<1$ \\
fpg-poly-0000000050 & 50 & \num{1.3485} & \num{1} & \num{8} & $<1$ \\
uniform-0000050-1 & 50 & \num{1.2665} & \num{1} & \num{1} & $<1$ \\
pntset-0000050 & 50 & \num{1.2845} & \num{3} & \num{7} & $<1$ \\
euro-night-0000050 & 50 & \num{1.2926} & \num{3} & \num{10} & $<1$ \\
us-night-0000050 & 50 & \num{1.2728} & \num{1} & \num{1} & $<1$ \\
eil51 & 51 & \num{1.3306} & \num{1} & \num{6} & $<1$ \\
berlin52 & 52 & \num{1.2693} & \num{3} & \num{8} & $<1$ \\
fpg-poly-0000000060 & 60 & \num{1.2829} & \num{1} & \num{8} & $<1$ \\
london-0000060 & 60 & \num{1.3246} & \num{2} & \num{9} & $<1$ \\
us-night-0000060 & 60 & \num{1.2564} & \num{1} & \num{8} & $<1$ \\
stars-0000060 & 60 & \num{1.3449} & \num{1} & \num{8} & $<1$ \\
uniform-0000060-2 & 60 & \num{1.3171} & \num{3} & \num{5} & $<1$ \\
pntset-0000060 & 60 & \num{1.3117} & \num{1} & \num{6} & $<1$ \\
uniform-0000060-1 & 60 & \num{1.3149} & \num{1} & \num{6} & $<1$ \\
euro-night-0000060 & 60 & \num{1.2973} & \num{1} & \num{7} & $<1$ \\
uniform-0000070-2 & 70 & \num{1.2880} & \num{1} & \num{1} & $<1$ \\
st70 & 70 & \num{1.3002} & \num{2} & \num{6} & $<1$ \\
stars-0000070 & 70 & \num{1.3375} & \num{5} & \num{16} & $<1$ \\
london-0000070 & 70 & \num{1.2838} & \num{4} & \num{11} & $<1$ \\
euro-night-0000070 & 70 & \num{1.3233} & \num{3} & \num{9} & $<1$ \\
fpg-poly-0000000070 & 70 & \num{1.2979} & \num{1} & \num{7} & $<1$ \\
pntset-0000070 & 70 & \num{1.3298} & \num{1} & \num{8} & $<1$ \\
uniform-0000070-1 & 70 & \num{1.2803} & \num{4} & \num{7} & $<1$ \\
us-night-0000070 & 70 & \num{1.2949} & \num{2} & \num{10} & $<1$ \\
pr76 & 76 & \num{1.3013} & \num{3} & \num{18} & $<1$ \\
eil76 & 76 & \num{1.3407} & \num{1} & \num{1} & $<1$ \\
pntset-0000080 & 80 & \num{1.2983} & \num{1} & \num{6} & $<1$ \\
us-night-0000080 & 80 & \num{1.3090} & \num{1} & \num{6} & $<1$ \\
uniform-0000080-1 & 80 & \num{1.3169} & \num{4} & \num{9} & $<1$ \\
london-0000080 & 80 & \num{1.3399} & \num{3} & \num{9} & $<1$ \\
stars-0000080 & 80 & \num{1.3066} & \num{3} & \num{10} & $<1$ \\
fpg-poly-0000000080 & 80 & \num{1.3236} & \num{3} & \num{5} & $<1$ \\
euro-night-0000080 & 80 & \num{1.3598} & \num{3} & \num{9} & $<1$ \\
uniform-0000080-2 & 80 & \num{1.3059} & \num{1} & \num{6} & $<1$ \\
uniform-0000090-1 & 90 & \num{1.3033} & \num{2} & \num{9} & $<1$ \\
euro-night-0000090 & 90 & \num{1.3374} & \num{1} & \num{7} & $<1$ \\
uniform-0000090-2 & 90 & \num{1.3021} & \num{3} & \num{10} & $<1$ \\
stars-0000090 & 90 & \num{1.2900} & \num{1} & \num{5} & $<1$ \\
fpg-poly-0000000090 & 90 & \num{1.3479} & \num{3} & \num{12} & $<1$ \\
us-night-0000090 & 90 & \num{1.3481} & \num{4} & \num{12} & $<1$ \\
pntset-0000090 & 90 & \num{1.2975} & \num{3} & \num{4} & $<1$ \\
london-0000090 & 90 & \num{1.3206} & \num{4} & \num{9} & $<1$ \\
gr96 & 96 & \num{1.3310} & \num{1} & \num{1} & $<1$ \\
rat99 & 99 & \num{1.3280} & \num{4} & \num{10} & $<1$ \\
fpg-poly-0000000100 & 100 & \num{1.2652} & \num{5} & \num{13} & $<1$ \\
kroD100 & 100 & \num{1.3159} & \num{1} & \num{6} & $<1$ \\
rd100 & 100 & \num{1.3328} & \num{2} & \num{10} & $<1$ \\
kroE100 & 100 & \num{1.3416} & \num{2} & \num{13} & $<1$ \\
kroA100 & 100 & \num{1.3386} & \num{1} & \num{1} & $<1$ \\
stars-0000100 & 100 & \num{1.3448} & \num{1} & \num{16} & $<1$ \\
us-night-0000100 & 100 & \num{1.3176} & \num{5} & \num{10} & $<1$ \\
uniform-0000100-1 & 100 & \num{1.3335} & \num{1} & \num{1} & $<1$ \\
london-0000100 & 100 & \num{1.3387} & \num{3} & \num{6} & $<1$ \\
pntset-0000100 & 100 & \num{1.3420} & \num{1} & \num{12} & $<1$ \\
uniform-0000100-2 & 100 & \num{1.3701} & \num{1} & \num{6} & $<1$ \\
euro-night-0000100 & 100 & \num{1.3067} & \num{4} & \num{10} & $<1$ \\
kroC100 & 100 & \num{1.3263} & \num{1} & \num{7} & $<1$ \\
kroB100 & 100 & \num{1.2910} & \num{3} & \num{7} & $<1$ \\
eil101 & 101 & \num{1.4142} & \num{1} & \num{8} & $<1$ \\
lin105 & 105 & \num{1.3116} & \num{6} & \num{15} & $<1$ \\
pr107 & 107 & \num{1.2504} & \num{41} & \num{45} & $<1$ \\
pr124 & 124 & \num{1.3077} & \num{5} & \num{12} & $<1$ \\
bier127 & 127 & \num{1.2999} & \num{3} & \num{9} & $<1$ \\
ch130 & 130 & \num{1.3535} & \num{1} & \num{24} & $<1$ \\
xqf131 & 131 & \num{1.3720} & \num{3} & \num{7} & $<1$ \\
pr136 & 136 & \num{1.4037} & \num{1} & \num{2} & $<1$ \\
gr137 & 137 & \num{1.2732} & \num{3} & \num{8} & $<1$ \\
pr144 & 144 & \num{1.2469} & \num{40} & \num{67} & $<1$ \\
kroA150 & 150 & \num{1.3386} & \num{1} & \num{1} & $<1$ \\
kroB150 & 150 & \num{1.3217} & \num{3} & \num{9} & $<1$ \\
ch150 & 150 & \num{1.3032} & \num{6} & \num{12} & $<1$ \\
pr152 & 152 & \num{1.2770} & \num{71} & \num{141} & $<1$ \\
u159 & 159 & \num{1.4142} & \num{1} & \num{8} & $<1$ \\
rat195 & 195 & \num{1.3436} & \num{2} & \num{10} & $<1$ \\
d198 & 198 & \num{1.4142} & \num{1} & \num{1} & $<1$ \\
stars-0000200 & 200 & \num{1.3574} & \num{2} & \num{8} & $<1$ \\
uniform-0000200-1 & 200 & \num{1.3563} & \num{2} & \num{9} & $<1$ \\
paris-0000200 & 200 & \num{1.3278} & \num{5} & \num{11} & $<1$ \\
us-night-0000200 & 200 & \num{1.3354} & \num{5} & \num{16} & $<1$ \\
kroB200 & 200 & \num{1.3804} & \num{1} & \num{1} & $<1$ \\
fpg-poly-0000000200 & 200 & \num{1.3144} & \num{7} & \num{13} & $<1$ \\
pntset-0000200 & 200 & \num{1.3582} & \num{1} & \num{6} & $<1$ \\
uniform-0000200-2 & 200 & \num{1.3560} & \num{1} & \num{12} & $<1$ \\
kroA200 & 200 & \num{1.3863} & \num{1} & \num{10} & $<1$ \\
euro-night-0000200 & 200 & \num{1.3147} & \num{5} & \num{14} & $<1$ \\
gr202 & 202 & \num{1.3184} & \num{2} & \num{7} & $<1$ \\
ts225 & 225 & \num{1.4142} & \num{1} & \num{10} & $<1$ \\
tsp225 & 225 & \num{1.4140} & \num{1} & \num{5} & $<1$ \\
pr226 & 226 & \num{1.4142} & \num{1} & \num{8} & $<1$ \\
gr229 & 229 & \num{1.3382} & \num{9} & \num{18} & $<1$ \\
xqg237 & 237 & \num{1.3720} & \num{8} & \num{21} & $<1$ \\
gil262 & 262 & \num{1.3209} & \num{6} & \num{14} & $<1$ \\
pr264 & 264 & \num{1.3868} & \num{16} & \num{38} & $<1$ \\
a280 & 279 & \num{1.4142} & \num{1} & \num{5} & $<1$ \\
pr299 & 299 & \num{1.4142} & \num{1} & \num{3} & $<1$ \\
stars-0000300 & 300 & \num{1.3392} & \num{18} & \num{34} & $<1$ \\
uniform-0000300-2 & 300 & \num{1.3632} & \num{6} & \num{12} & $<1$ \\
uniform-0000300-1 & 300 & \num{1.3317} & \num{6} & \num{15} & $<1$ \\
us-night-0000300 & 300 & \num{1.3521} & \num{1} & \num{7} & $<1$ \\
fpg-poly-0000000300 & 300 & \num{1.3500} & \num{4} & \num{19} & $<1$ \\
pntset-0000300 & 300 & \num{1.3210} & \num{5} & \num{12} & $<1$ \\
paris-0000300 & 300 & \num{1.3577} & \num{2} & \num{13} & $<1$ \\
euro-night-0000300 & 300 & \num{1.3319} & \num{25} & \num{34} & $<1$ \\
linhp318 & 318 & \num{1.3145} & \num{22} & \num{37} & $<1$ \\
lin318 & 318 & \num{1.3145} & \num{22} & \num{37} & $<1$ \\
pma343 & 343 & \num{1.3473} & \num{8} & \num{25} & $<1$ \\
pka379 & 379 & \num{1.3473} & \num{5} & \num{26} & $<1$ \\
bcl380 & 380 & \num{1.3720} & \num{15} & \num{28} & $<1$ \\
pbl395 & 395 & \num{1.3473} & \num{6} & \num{21} & $<1$ \\
euro-night-0000400 & 400 & \num{1.3486} & \num{23} & \num{32} & $<1$ \\
stars-0000400 & 400 & \num{1.3601} & \num{7} & \num{18} & $<1$ \\
uniform-0000400-2 & 400 & \num{1.3793} & \num{1} & \num{7} & $<1$ \\
rd400 & 400 & \num{1.3489} & \num{10} & \num{21} & $<1$ \\
fpg-poly-0000000400 & 400 & \num{1.3292} & \num{26} & \num{39} & $<1$ \\
pntset-0000400 & 400 & \num{1.3443} & \num{5} & \num{11} & $<1$ \\
us-night-0000400 & 400 & \num{1.3489} & \num{11} & \num{27} & $<1$ \\
paris-0000400 & 400 & \num{1.3392} & \num{4} & \num{15} & $<1$ \\
uniform-0000400-1 & 400 & \num{1.3762} & \num{1} & \num{7} & $<1$ \\
pbk411 & 411 & \num{1.4123} & \num{6} & \num{18} & $<1$ \\
fl417 & 417 & \num{1.4142} & \num{1} & \num{1} & $<1$ \\
pbn423 & 423 & \num{1.3868} & \num{15} & \num{30} & $<1$ \\
gr431 & 431 & \num{1.3409} & \num{7} & \num{16} & $<1$ \\
pbm436 & 436 & \num{1.3868} & \num{12} & \num{24} & $<1$ \\
pr439 & 439 & \num{1.4142} & \num{18} & \num{29} & $<1$ \\
pcb442 & 442 & \num{1.4142} & \num{1} & \num{7} & $<1$ \\
d493 & 493 & \num{1.3416} & \num{9} & \num{17} & $<1$ \\
us-night-0000500 & 500 & \num{1.3324} & \num{17} & \num{35} & $<1$ \\
uniform-0000500-1 & 500 & \num{1.3420} & \num{9} & \num{20} & $<1$ \\
pntset-0000500 & 500 & \num{1.3640} & \num{4} & \num{26} & $<1$ \\
paris-0000500 & 500 & \num{1.3574} & \num{12} & \num{27} & $<1$ \\
stars-0000500 & 500 & \num{1.3405} & \num{6} & \num{15} & $<1$ \\
fpg-poly-0000000500 & 500 & \num{1.3463} & \num{1} & \num{1} & $<1$ \\
euro-night-0000500 & 500 & \num{1.3611} & \num{16} & \num{28} & $<1$ \\
uniform-0000500-2 & 500 & \num{1.3690} & \num{4} & \num{13} & $<1$ \\
ali535 & 506 & \num{1.3361} & \num{27} & \num{46} & $<1$ \\
att532 & 532 & \num{1.3630} & \num{4} & \num{12} & $<1$ \\
u574 & 574 & \num{1.3814} & \num{4} & \num{18} & $<1$ \\
rat575 & 575 & \num{1.3558} & \num{6} & \num{19} & $<1$ \\
euro-night-0000600 & 600 & \num{1.3716} & \num{31} & \num{50} & $<1$ \\
uniform-0000600-1 & 600 & \num{1.3625} & \num{3} & \num{7} & $<1$ \\
paris-0000600 & 600 & \num{1.3797} & \num{1} & \num{1} & $<1$ \\
pntset-0000600 & 600 & \num{1.3385} & \num{8} & \num{19} & $<1$ \\
us-night-0000600 & 600 & \num{1.3999} & \num{1} & \num{20} & $<1$ \\
stars-0000600 & 600 & \num{1.3642} & \num{11} & \num{31} & $<1$ \\
uniform-0000600-2 & 600 & \num{1.3555} & \num{10} & \num{28} & $<1$ \\
fpg-poly-0000000600 & 600 & \num{1.3541} & \num{27} & \num{36} & $<1$ \\
p654 & 654 & \num{1.4142} & \num{1} & \num{8} & \num{1} \\
d657 & 657 & \num{1.3473} & \num{28} & \num{43} & $<1$ \\
xql662 & 662 & \num{1.3552} & \num{37} & \num{58} & $<1$ \\
gr666 & 666 & \num{1.3409} & \num{16} & \num{27} & $<1$ \\
stars-0000700 & 700 & \num{1.3536} & \num{15} & \num{28} & $<1$ \\
uniform-0000700-2 & 700 & \num{1.3792} & \num{1} & \num{4} & $<1$ \\
euro-night-0000700 & 700 & \num{1.3790} & \num{1} & \num{3} & $<1$ \\
fpg-poly-0000000700 & 700 & \num{1.3664} & \num{12} & \num{37} & $<1$ \\
uniform-0000700-1 & 700 & \num{1.3838} & \num{1} & \num{12} & $<1$ \\
us-night-0000700 & 700 & \num{1.3763} & \num{10} & \num{30} & $<1$ \\
pntset-0000700 & 700 & \num{1.3510} & \num{7} & \num{17} & $<1$ \\
paris-0000700 & 700 & \num{1.3303} & \num{17} & \num{30} & $<1$ \\
rbx711 & 711 & \num{1.4000} & \num{9} & \num{29} & $<1$ \\
u724 & 724 & \num{1.4142} & \num{1} & \num{1} & $<1$ \\
rbu737 & 737 & \num{1.4142} & \num{1} & \num{8} & $<1$ \\
rat783 & 783 & \num{1.3613} & \num{5} & \num{17} & $<1$ \\
stars-0000800 & 800 & \num{1.3572} & \num{11} & \num{22} & $<1$ \\
euro-night-0000800 & 800 & \num{1.3425} & \num{44} & \num{63} & $<1$ \\
uniform-0000800-2 & 800 & \num{1.3866} & \num{1} & \num{5} & $<1$ \\
fpg-poly-0000000800 & 800 & \num{1.3795} & \num{25} & \num{44} & $<1$ \\
pntset-0000800 & 800 & \num{1.3942} & \num{1} & \num{4} & $<1$ \\
uniform-0000800-1 & 800 & \num{1.3701} & \num{13} & \num{28} & $<1$ \\
paris-0000800 & 800 & \num{1.3504} & \num{1} & \num{8} & $<1$ \\
us-night-0000800 & 800 & \num{1.3543} & \num{20} & \num{33} & $<1$ \\
dkg813 & 813 & \num{1.4000} & \num{25} & \num{39} & $<1$ \\
us-night-0000900 & 900 & \num{1.3611} & \num{38} & \num{71} & $<1$ \\
uniform-0000900-1 & 900 & \num{1.3636} & \num{5} & \num{17} & $<1$ \\
euro-night-0000900 & 900 & \num{1.3655} & \num{52} & \num{70} & $<1$ \\
stars-0000900 & 900 & \num{1.3352} & \num{28} & \num{44} & $<1$ \\
uniform-0000900-2 & 900 & \num{1.3625} & \num{4} & \num{14} & $<1$ \\
paris-0000900 & 900 & \num{1.3599} & \num{19} & \num{38} & $<1$ \\
fpg-poly-0000000900 & 900 & \num{1.3551} & \num{38} & \num{51} & $<1$ \\
pntset-0000900 & 900 & \num{1.3897} & \num{1} & \num{7} & $<1$ \\
lim963 & 963 & \num{1.4142} & \num{1} & \num{5} & $<1$ \\
pbd984 & 984 & \num{1.4142} & \num{1} & \num{5} & $<1$ \\
skylake-0001000 & 1000 & \num{1.3660} & \num{12} & \num{26} & $<1$ \\
dsj1000 & 1000 & \num{1.3471} & \num{49} & \num{66} & $<1$ \\
us-night-0001000 & 1000 & \num{1.3600} & \num{1} & \num{5} & $<1$ \\
uniform-0001000-2 & 1000 & \num{1.3782} & \num{1} & \num{5} & $<1$ \\
fpg-poly-0000001000 & 1000 & \num{1.3498} & \num{26} & \num{44} & $<1$ \\
uniform-0001000-1 & 1000 & \num{1.3699} & \num{19} & \num{36} & $<1$ \\
euro-night-0001000 & 1000 & \num{1.3750} & \num{1} & \num{30} & $<1$ \\
pntset-0001000 & 1000 & \num{1.3887} & \num{1} & \num{25} & $<1$ \\
paris-0001000 & 1000 & \num{1.3721} & \num{13} & \num{27} & $<1$ \\
pr1002 & 1002 & \num{1.4142} & \num{1} & \num{5} & $<1$ \\
u1060 & 1060 & \num{1.4142} & \num{1} & \num{4} & $<1$ \\
xit1083 & 1083 & \num{1.4142} & \num{1} & \num{16} & $<1$ \\
vm1084 & 1084 & \num{1.3969} & \num{1} & \num{7} & $<1$ \\
fpg-poly-0000001100 & 1100 & \num{1.3605} & \num{29} & \num{53} & $<1$ \\
pcb1173 & 1173 & \num{1.3959} & \num{13} & \num{36} & $<1$ \\
fpg-poly-0000001200 & 1200 & \num{1.3626} & \num{18} & \num{32} & $<1$ \\
d1291 & 1291 & \num{1.4142} & \num{1} & \num{8} & $<1$ \\
fpg-poly-0000001300 & 1300 & \num{1.3702} & \num{53} & \num{74} & \num{1} \\
rl1304 & 1304 & \num{1.3844} & \num{111} & \num{140} & \num{3} \\
rl1323 & 1323 & \num{1.4137} & \num{1} & \num{12} & $<1$ \\
dka1376 & 1376 & \num{1.4142} & \num{1} & \num{10} & $<1$ \\
nrw1379 & 1379 & \num{1.3731} & \num{16} & \num{32} & $<1$ \\
dca1389 & 1389 & \num{1.4000} & \num{46} & \num{65} & \num{2} \\
fpg-poly-0000001400 & 1400 & \num{1.3959} & \num{1} & \num{2} & $<1$ \\
fl1400 & 1400 & \num{1.4142} & \num{1} & \num{1} & $<1$ \\
u1432 & 1432 & \num{1.4142} & \num{8} & \num{18} & \num{3} \\
dja1436 & 1436 & \num{1.4142} & \num{1} & \num{12} & $<1$ \\
icw1483 & 1483 & \num{1.4142} & \num{1} & \num{31} & \num{1} \\
fra1488 & 1488 & \num{1.4142} & \num{1} & \num{6} & $<1$ \\
fpg-poly-0000001500 & 1500 & \num{1.3545} & \num{63} & \num{96} & \num{2} \\
fl1577 & 1577 & \num{1.4142} & \num{1} & \num{1} & \num{5} \\
rbv1583 & 1583 & \num{1.4000} & \num{44} & \num{73} & \num{3} \\
rby1599 & 1599 & \num{1.4000} & \num{36} & \num{65} & \num{3} \\
fpg-poly-0000001600 & 1600 & \num{1.3705} & \num{58} & \num{85} & \num{2} \\
fnb1615 & 1615 & \num{1.4142} & \num{1} & \num{12} & $<1$ \\
d1655 & 1655 & \num{1.4142} & \num{1} & \num{1} & $<1$ \\
fpg-poly-0000001700 & 1700 & \num{1.3457} & \num{60} & \num{85} & \num{3} \\
vm1748 & 1748 & \num{1.3964} & \num{1} & \num{39} & \num{3} \\
djc1785 & 1785 & \num{1.4142} & \num{1} & \num{17} & $<1$ \\
fpg-poly-0000001800 & 1800 & \num{1.3742} & \num{115} & \num{145} & \num{5} \\
u1817 & 1817 & \num{1.4142} & \num{1} & \num{1} & $<1$ \\
rl1889 & 1889 & \num{1.4142} & \num{1} & \num{10} & \num{1} \\
fpg-poly-0000001900 & 1900 & \num{1.3768} & \num{96} & \num{121} & \num{5} \\
dcc1911 & 1911 & \num{1.4142} & \num{1} & \num{18} & $<1$ \\
dkd1973 & 1973 & \num{1.3540} & \num{92} & \num{123} & \num{8} \\
euro-night-0002000 & 2000 & \num{1.3518} & \num{145} & \num{193} & \num{8} \\
pntset-0002000 & 2000 & \num{1.3812} & \num{1} & \num{9} & $<1$ \\
uniform-0002000-1 & 2000 & \num{1.3700} & \num{16} & \num{38} & \num{1} \\
uniform-0002000-2 & 2000 & \num{1.3644} & \num{18} & \num{42} & \num{1} \\
paris-0002000 & 2000 & \num{1.3754} & \num{53} & \num{83} & \num{3} \\
skylake-0002000 & 2000 & \num{1.3680} & \num{29} & \num{62} & \num{2} \\
fpg-poly-0000002000 & 2000 & \num{1.3631} & \num{40} & \num{70} & \num{3} \\
us-night-0002000 & 2000 & \num{1.3745} & \num{63} & \num{105} & \num{5} \\
djb2036 & 2036 & \num{1.4142} & \num{1} & \num{24} & \num{1} \\
dcb2086 & 2086 & \num{1.4142} & \num{1} & \num{6} & $<1$ \\
fpg-poly-0000002100 & 2100 & \num{1.3544} & \num{125} & \num{177} & \num{9} \\
d2103 & 2103 & \num{1.4142} & \num{1} & \num{155} & \num{39} \\
bva2144 & 2144 & \num{1.4142} & \num{1} & \num{9} & $<1$ \\
u2152 & 2152 & \num{1.4142} & \num{61} & \num{109} & \num{14} \\
xqc2175 & 2175 & \num{1.4142} & \num{1} & \num{9} & $<1$ \\
fpg-poly-0000002200 & 2200 & \num{1.3666} & \num{71} & \num{113} & \num{5} \\
bck2217 & 2217 & \num{1.4142} & \num{1} & \num{8} & $<1$ \\
fpg-poly-0000002300 & 2300 & \num{1.3732} & \num{128} & \num{170} & \num{9} \\
xpr2308 & 2308 & \num{1.4142} & \num{1} & \num{5} & $<1$ \\
u2319 & 2319 & \num{1.4142} & \num{1} & \num{1} & $<1$ \\
ley2323 & 2323 & \num{1.3720} & \num{43} & \num{102} & \num{7} \\
dea2382 & 2382 & \num{1.4142} & \num{1} & \num{45} & \num{3} \\
pr2392 & 2392 & \num{1.4142} & \num{1} & \num{3} & $<1$ \\
fpg-poly-0000002400 & 2400 & \num{1.3797} & \num{81} & \num{122} & \num{7} \\
rbw2481 & 2481 & \num{1.4142} & \num{1} & \num{24} & \num{2} \\
fpg-poly-0000002500 & 2500 & \num{1.3656} & \num{109} & \num{150} & \num{9} \\
pds2566 & 2566 & \num{1.3686} & \num{95} & \num{127} & \num{13} \\
mlt2597 & 2597 & \num{1.3868} & \num{130} & \num{182} & \num{18} \\
fpg-poly-0000002600 & 2600 & \num{1.3647} & \num{114} & \num{154} & \num{10} \\
fpg-poly-0000002700 & 2700 & \num{1.3729} & \num{125} & \num{172} & \num{12} \\
bch2762 & 2762 & \num{1.4142} & \num{1} & \num{12} & \num{1} \\
fpg-poly-0000002800 & 2800 & \num{1.3670} & \num{109} & \num{147} & \num{11} \\
irw2802 & 2802 & \num{1.4142} & \num{1} & \num{9} & $<1$ \\
lsm2854 & 2854 & \num{1.4142} & \num{1} & \num{28} & \num{2} \\
fpg-poly-0000002900 & 2900 & \num{1.3945} & \num{66} & \num{123} & \num{10} \\
dbj2924 & 2924 & \num{1.4142} & \num{1} & \num{7} & $<1$ \\
xva2993 & 2993 & \num{1.4142} & \num{1} & \num{8} & $<1$ \\
fpg-poly-0000003000 & 3000 & \num{1.3758} & \num{133} & \num{177} & \num{16} \\
uniform-0003000-1 & 3000 & \num{1.3756} & \num{46} & \num{77} & \num{6} \\
euro-night-0003000 & 3000 & \num{1.3746} & \num{194} & \num{247} & \num{22} \\
pntset-0003000 & 3000 & \num{1.3818} & \num{36} & \num{71} & \num{5} \\
skylake-0003000 & 3000 & \num{1.3839} & \num{53} & \num{93} & \num{7} \\
uniform-0003000-2 & 3000 & \num{1.3864} & \num{23} & \num{57} & \num{4} \\
us-night-0003000 & 3000 & \num{1.3664} & \num{169} & \num{243} & \num{23} \\
paris-0003000 & 3000 & \num{1.3655} & \num{118} & \num{173} & \num{15} \\
pcb3038 & 3038 & \num{1.4142} & \num{1} & \num{14} & $<1$ \\
pia3056 & 3056 & \num{1.4142} & \num{1} & \num{32} & \num{3} \\
dke3097 & 3097 & \num{1.4142} & \num{1} & \num{46} & \num{3} \\
fpg-poly-0000003100 & 3100 & \num{1.3549} & \num{112} & \num{161} & \num{15} \\
lsn3119 & 3119 & \num{1.4142} & \num{1} & \num{9} & $<1$ \\
lta3140 & 3140 & \num{1.4142} & \num{1} & \num{6} & $<1$ \\
fpg-poly-0000003200 & 3200 & \num{1.3695} & \num{127} & \num{170} & \num{17} \\
fdp3256 & 3256 & \num{1.4142} & \num{1} & \num{7} & $<1$ \\
beg3293 & 3293 & \num{1.3955} & \num{297} & \num{462} & \num{80} \\
fpg-poly-0000003300 & 3300 & \num{1.3903} & \num{115} & \num{152} & \num{17} \\
dhb3386 & 3386 & \num{1.4142} & \num{1} & \num{21} & \num{2} \\
fpg-poly-0000003400 & 3400 & \num{1.3735} & \num{126} & \num{175} & \num{19} \\
fpg-poly-0000003500 & 3500 & \num{1.3795} & \num{157} & \num{204} & \num{24} \\
fpg-poly-0000003600 & 3600 & \num{1.3732} & \num{146} & \num{205} & \num{25} \\
fjs3649 & 3649 & \num{1.4142} & \num{1} & \num{5} & $<1$ \\
fjr3672 & 3672 & \num{1.4142} & \num{1} & \num{5} & $<1$ \\
dlb3694 & 3694 & \num{1.4142} & \num{1} & \num{21} & \num{2} \\
fpg-poly-0000003700 & 3700 & \num{1.3741} & \num{221} & \num{284} & \num{36} \\
ltb3729 & 3729 & \num{1.4142} & \num{1} & \num{6} & $<1$ \\
fl3795 & 3795 & \num{1.4142} & \num{1} & \num{1} & \num{52} \\
fpg-poly-0000003800 & 3800 & \num{1.4038} & \num{1} & \num{98} & \num{8} \\
xqe3891 & 3891 & \num{1.4142} & \num{1} & \num{5} & $<1$ \\
fpg-poly-0000003900 & 3900 & \num{1.3691} & \num{133} & \num{191} & \num{27} \\
xua3937 & 3937 & \num{1.4142} & \num{1} & \num{2} & $<1$ \\
dkc3938 & 3938 & \num{1.4142} & \num{1} & \num{20} & \num{2} \\
dkf3954 & 3954 & \num{1.4142} & \num{1} & \num{98} & \num{7} \\
paris-0004000 & 4000 & \num{1.3629} & \num{147} & \num{214} & \num{30} \\
uniform-0004000-2 & 4000 & \num{1.3942} & \num{51} & \num{99} & \num{13} \\
us-night-0004000 & 4000 & \num{1.3902} & \num{176} & \num{231} & \num{38} \\
fpg-poly-0000004000 & 4000 & \num{1.3658} & \num{223} & \num{351} & \num{55} \\
euro-night-0004000 & 4000 & \num{1.3761} & \num{171} & \num{215} & \num{38} \\
pntset-0004000 & 4000 & \num{1.3874} & \num{51} & \num{109} & \num{14} \\
skylake-0004000 & 4000 & \num{1.3970} & \num{1} & \num{21} & \num{2} \\
uniform-0004000-1 & 4000 & \num{1.3668} & \num{49} & \num{93} & \num{11} \\
fpg-poly-0000004100 & 4100 & \num{1.3789} & \num{165} & \num{209} & \num{33} \\
fpg-poly-0000004200 & 4200 & \num{1.3848} & \num{149} & \num{202} & \num{37} \\
fpg-poly-0000004300 & 4300 & \num{1.3634} & \num{225} & \num{298} & \num{50} \\
bgb4355 & 4355 & \num{1.4142} & \num{1} & \num{9} & $<1$ \\
bgd4396 & 4396 & \num{1.4142} & \num{1} & \num{30} & \num{3} \\
fpg-poly-0000004400 & 4400 & \num{1.3998} & \num{131} & \num{181} & \num{33} \\
frv4410 & 4410 & \num{1.4142} & \num{1} & \num{35} & \num{3} \\
fnl4461 & 4461 & \num{1.4084} & \num{15} & \num{34} & \num{5} \\
bgf4475 & 4475 & \num{1.3868} & \num{203} & \num{269} & \num{81} \\
fpg-poly-0000004500 & 4500 & \num{1.3754} & \num{171} & \num{239} & \num{46} \\
fpg-poly-0000004600 & 4600 & \num{1.3669} & \num{217} & \num{275} & \num{60} \\
fpg-poly-0000004700 & 4700 & \num{1.3901} & \num{195} & \num{251} & \num{52} \\
fpg-poly-0000004800 & 4800 & \num{1.3859} & \num{171} & \num{234} & \num{52} \\
fpg-poly-0000004900 & 4900 & \num{1.3983} & \num{1} & \num{225} & \num{46} \\
xqd4966 & 4966 & \num{1.4142} & \num{1} & \num{8} & $<1$ \\
pntset-0005000 & 5000 & \num{1.3846} & \num{50} & \num{104} & \num{20} \\
skylake-0005000 & 5000 & \num{1.3791} & \num{109} & \num{182} & \num{38} \\
fpg-poly-0000005000 & 5000 & \num{1.3771} & \num{268} & \num{326} & \num{78} \\
uniform-0005000-1 & 5000 & \num{1.3871} & \num{53} & \num{97} & \num{18} \\
euro-night-0005000 & 5000 & \num{1.3880} & \num{278} & \num{338} & \num{84} \\
us-night-0005000 & 5000 & \num{1.3717} & \num{322} & \num{408} & \num{104} \\
uniform-0005000-2 & 5000 & \num{1.3760} & \num{68} & \num{122} & \num{25} \\
paris-0005000 & 5000 & \num{1.3818} & \num{171} & \num{237} & \num{54} \\
fqm5087 & 5087 & \num{1.4142} & \num{1} & \num{23} & \num{3} \\
fpg-poly-0000005100 & 5100 & \num{1.3692} & \num{233} & \num{296} & \num{74} \\
fpg-poly-0000005200 & 5200 & \num{1.3853} & \num{1} & \num{9} & \num{1} \\
fpg-poly-0000005300 & 5300 & \num{1.3687} & \num{243} & \num{314} & \num{84} \\
fpg-poly-0000005400 & 5400 & \num{1.3686} & \num{261} & \num{343} & \num{104} \\
fpg-poly-0000005500 & 5500 & \num{1.3844} & \num{292} & \num{365} & \num{106} \\
fea5557 & 5557 & \num{1.4142} & \num{1} & \num{16} & \num{3} \\
fpg-poly-0000005600 & 5600 & \num{1.4035} & \num{1} & \num{211} & \num{62} \\
fpg-poly-0000005700 & 5700 & \num{1.3937} & \num{218} & \num{267} & \num{84} \\
fpg-poly-0000005800 & 5800 & \num{1.3847} & \num{286} & \num{359} & \num{114} \\
fpg-poly-0000005900 & 5900 & \num{1.3864} & \num{192} & \num{258} & \num{82} \\
rl5915 & 5915 & \num{1.4142} & \num{1} & \num{34} & \num{9} \\
rl5934 & 5934 & \num{1.4142} & \num{1} & \num{7} & \num{3} \\
uniform-0006000-1 & 6000 & \num{1.3929} & \num{59} & \num{126} & \num{35} \\
skylake-0006000 & 6000 & \num{1.3940} & \num{106} & \num{188} & \num{55} \\
us-night-0006000 & 6000 & \num{1.3884} & \num{307} & \num{380} & \num{133} \\
uniform-0006000-2 & 6000 & \num{1.3832} & \num{56} & \num{126} & \num{32} \\
fpg-poly-0000006000 & 6000 & \num{1.3811} & \num{207} & \num{283} & \num{106} \\
euro-night-0006000 & 6000 & \num{1.3776} & \num{485} & \num{567} & \num{208} \\
world-0006000 & 6000 & \num{1.3805} & \num{500} & \num{591} & \num{237} \\
pntset-0006000 & 6000 & \num{1.3787} & \num{68} & \num{134} & \num{36} \\
fpg-poly-0000006100 & 6100 & \num{1.3940} & \num{130} & \num{188} & \num{64} \\
fpg-poly-0000006200 & 6200 & \num{1.3727} & \num{289} & \num{366} & \num{137} \\
fpg-poly-0000006300 & 6300 & \num{1.3880} & \num{273} & \num{363} & \num{145} \\
fpg-poly-0000006400 & 6400 & \num{1.3796} & \num{342} & \num{432} & \num{168} \\
fpg-poly-0000006500 & 6500 & \num{1.3693} & \num{355} & \num{452} & \num{206} \\
fpg-poly-0000006600 & 6600 & \num{1.3790} & \num{234} & \num{321} & \num{129} \\
fpg-poly-0000006700 & 6700 & \num{1.3947} & \num{192} & \num{260} & \num{116} \\
fpg-poly-0000006800 & 6800 & \num{1.3663} & \num{317} & \num{433} & \num{186} \\
xsc6880 & 6880 & \num{1.4142} & \num{1} & \num{26} & \num{4} \\
fpg-poly-0000006900 & 6900 & \num{1.3795} & \num{265} & \num{350} & \num{170} \\
uniform-0007000-1 & 7000 & \num{1.3672} & \num{55} & \num{133} & \num{43} \\
euro-night-0007000 & 7000 & \num{1.3844} & \num{422} & \num{522} & \num{263} \\
uniform-0007000-2 & 7000 & \num{1.3868} & \num{80} & \num{147} & \num{55} \\
skylake-0007000 & 7000 & \num{1.3854} & \num{169} & \num{253} & \num{104} \\
us-night-0007000 & 7000 & \num{1.3776} & \num{622} & \num{749} & \num{370} \\
pntset-0007000 & 7000 & \num{1.3919} & \num{61} & \num{144} & \num{53} \\
fpg-poly-0000007000 & 7000 & \num{1.3818} & \num{301} & \num{393} & \num{185} \\
world-0007000 & 7000 & \num{1.3699} & \num{557} & \num{667} & \num{330} \\
fpg-poly-0000007100 & 7100 & \num{1.3897} & \num{218} & \num{297} & \num{134} \\
bnd7168 & 7168 & \num{1.4142} & \num{1} & \num{27} & \num{5} \\
fpg-poly-0000007200 & 7200 & \num{1.3851} & \num{287} & \num{379} & \num{188} \\
fpg-poly-0000007300 & 7300 & \num{1.3786} & \num{346} & \num{443} & \num{248} \\
pla7397 & 7397 & \num{1.4142} & \num{1055} & \num{1098} & \num{6502} \\
fpg-poly-0000007400 & 7400 & \num{1.3812} & \num{560} & \num{656} & \num{357} \\
lap7454 & 7454 & \num{1.4142} & \num{1} & \num{49} & \num{9} \\
fpg-poly-0000007500 & 7500 & \num{1.3864} & \num{291} & \num{379} & \num{204} \\
fpg-poly-0000007600 & 7600 & \num{1.3859} & \num{312} & \num{407} & \num{224} \\
fpg-poly-0000007700 & 7700 & \num{1.3737} & \num{349} & \num{448} & \num{285} \\
fpg-poly-0000007800 & 7800 & \num{1.3790} & \num{302} & \num{391} & \num{227} \\
fpg-poly-0000007900 & 7900 & \num{1.3878} & \num{232} & \num{324} & \num{190} \\
fpg-poly-0000008000 & 8000 & \num{1.3996} & \num{327} & \num{454} & \num{279} \\
pntset-0008000 & 8000 & \num{1.3897} & \num{102} & \num{178} & \num{85} \\
us-night-0008000 & 8000 & \num{1.3791} & \num{595} & \num{694} & \num{444} \\
world-0008000 & 8000 & \num{1.3793} & \num{616} & \num{765} & \num{493} \\
uniform-0008000-2 & 8000 & \num{1.3924} & \num{66} & \num{138} & \num{64} \\
jupiter-0008000 & 8000 & \num{1.3700} & \num{611} & \num{723} & \num{449} \\
euro-night-0008000 & 8000 & \num{1.3909} & \num{432} & \num{547} & \num{386} \\
uniform-0008000-1 & 8000 & \num{1.3956} & \num{90} & \num{197} & \num{90} \\
fpg-poly-0000008100 & 8100 & \num{1.3880} & \num{353} & \num{454} & \num{286} \\
ida8197 & 8197 & \num{1.4142} & \num{1} & \num{10} & \num{2} \\
fpg-poly-0000008200 & 8200 & \num{1.3790} & \num{311} & \num{397} & \num{252} \\
fpg-poly-0000008300 & 8300 & \num{1.3718} & \num{331} & \num{464} & \num{294} \\
fpg-poly-0000008400 & 8400 & \num{1.3885} & \num{294} & \num{411} & \num{261} \\
fpg-poly-0000008500 & 8500 & \num{1.3823} & \num{281} & \num{485} & \num{324} \\
fpg-poly-0000008600 & 8600 & \num{1.3791} & \num{406} & \num{519} & \num{363} \\
fpg-poly-0000008700 & 8700 & \num{1.3777} & \num{328} & \num{450} & \num{328} \\
fpg-poly-0000008800 & 8800 & \num{1.4029} & \num{346} & \num{440} & \num{325} \\
fpg-poly-0000008900 & 8900 & \num{1.3835} & \num{387} & \num{525} & \num{399} \\
world-0009000 & 9000 & \num{1.3723} & \num{621} & \num{746} & \num{591} \\
jupiter-0009000 & 9000 & \num{1.3971} & \num{408} & \num{497} & \num{402} \\
euro-night-0009000 & 9000 & \num{1.3864} & \num{904} & \num{1245} & \num{1053} \\
fpg-poly-0000009000 & 9000 & \num{1.3911} & \num{308} & \num{403} & \num{303} \\
pntset-0009000 & 9000 & \num{1.3883} & \num{75} & \num{131} & \num{74} \\
us-night-0009000 & 9000 & \num{1.3802} & \num{468} & \num{568} & \num{459} \\
uniform-0009000-2 & 9000 & \num{1.3896} & \num{87} & \num{179} & \num{111} \\
uniform-0009000-1 & 9000 & \num{1.3895} & \num{90} & \num{191} & \num{122} \\
fpg-poly-0000009100 & 9100 & \num{1.3823} & \num{268} & \num{411} & \num{318} \\
fpg-poly-0000009200 & 9200 & \num{1.3812} & \num{341} & \num{466} & \num{362} \\
fpg-poly-0000009300 & 9300 & \num{1.3922} & \num{299} & \num{401} & \num{331} \\
fpg-poly-0000009400 & 9400 & \num{1.3789} & \num{416} & \num{509} & \num{468} \\
fpg-poly-0000009500 & 9500 & \num{1.3927} & \num{369} & \num{491} & \num{420} \\
fpg-poly-0000009600 & 9600 & \num{1.3951} & \num{307} & \num{401} & \num{348} \\
dga9698 & 9698 & \num{1.4142} & \num{1} & \num{16} & \num{4} \\
fpg-poly-0000009700 & 9700 & \num{1.3712} & \num{417} & \num{608} & \num{548} \\
fpg-poly-0000009800 & 9800 & \num{1.3896} & \num{389} & \num{487} & \num{447} \\
fpg-poly-0000009900 & 9900 & \num{1.3836} & \num{512} & \num{646} & \num{687} \\
euro-night-0010000 & 10000 & \num{1.3868} & \num{462} & \num{601} & \num{577} \\
uniform-0010000-2 & 10000 & \num{1.3756} & \num{107} & \num{197} & \num{138} \\
world-0010000 & 10000 & \num{1.3973} & \num{519} & \num{651} & \num{727} \\
uniform-0010000-1 & 10000 & \num{1.3811} & \num{133} & \num{244} & \num{195} \\
pntset-0010000 & 10000 & \num{1.4079} & \num{90} & \num{187} & \num{136} \\
jupiter-0010000 & 10000 & \num{1.3888} & \num{618} & \num{786} & \num{773} \\
fpg-poly-0000010000 & 10000 & \num{1.3861} & \num{455} & \num{584} & \num{556} \\
us-night-0010000 & 10000 & \num{1.3859} & \num{397} & \num{515} & \num{502} \\
\hline

\caption{Solutions and runtimes for \emph{public} benchmark set together with the number 
         of full and sampled dilation computations that are the bottleneck in the exact algorithm. 
         The presented data was obtained using the binary solver with an improved delaunay triangulation as an initial solution.}
\label{tab:public-comparison}
\end{longtable}